\documentclass[11pt]{article}

\usepackage[english]{babel} 
\usepackage[T1]{fontenc}    
\usepackage{amsmath,amsfonts,amsthm,amssymb}
\usepackage{bbm}
\usepackage[bookmarksnumbered=true]{hyperref}       
\hypersetup{
	colorlinks   = true,    
	urlcolor     = black,   
	linkcolor    = black,   
	citecolor    = black,   
}
\usepackage[a4paper, total={6in, 9in}]{geometry}
\usepackage[nameinlink]{cleveref}     
\usepackage[final]{microtype}
\numberwithin{equation}{section}

\newtheorem{theorem}{Theorem}[section]
\newtheorem{T}[theorem]{Theorem}

\newtheorem{Le}[theorem]{Lemma}
\newtheorem{Prop}[theorem]{Proposition}
\theoremstyle{definition}
\newtheorem{D}[theorem]{Definition}

\newtheorem{A}[theorem]{Assumption}

\newtheorem{Rs}[theorem]{Remarks}

\theoremstyle{remark}

\crefname{T}{Theorem}{Theorems}
\Crefname{T}{Theorem}{Theorems}
\crefname{C}{Corollary}{Corollaries}
\Crefname{C}{Corollary}{Corollaries}
\crefname{Con}{conjecture}{Conjectures}
\Crefname{Con}{Conjecture}{Conjectures}
\crefname{Le}{Lemma}{Lemmas}
\Crefname{Le}{Lemma}{Lemmas}
\crefname{Prop}{Proposition}{Propositions}
\Crefname{Prop}{Proposition}{Propositions}
\crefname{D}{Definition}{Definitions}
\Crefname{D}{Definition}{Definitions}
\crefname{E}{Example}{Examples}
\Crefname{A}{Assumption}{Assumptions}
\crefname{A}{Assumption}{Assumptions}
\Crefname{N}{Notation}{Notations}
\crefname{N}{Notation}{Notations}
\crefname{cond}{Condition}{Conditions}
\crefname{Appendix}{Appendix}{Appendices}

\newcommand\cE{{\mathcal E}}
\newcommand\cF{{\mathcal F}}

\newcommand\cH{{\mathcal H}}

\newcommand\cN{{\mathcal N}}

\newcommand\cR{{\mathcal R}}

\newcommand\frh{{\mathfrak H}}

\newcommand{\bQ}{\mathbb{Q}}
\newcommand{\bR}{\mathbb{R}}
\newcommand{\bX}{\mathbb{X}}
\newcommand{\bC}{\mathbb{C}}
\newcommand{\bH}{\mathbb{H}}
\newcommand{\bZ}{\mathbb{Z}}
\newcommand{\bS}{\mathbb{S}}

\newcommand{\ee}{{\rm e}}
\newcommand{\ii}{{\rm i}}
\newcommand{\dd}{{\rm d}}

\begin{document}
\renewcommand{\theenumi}{\roman{enumi})}
\renewcommand{\thefootnote}{\fnsymbol{footnote}}

\title{The Huang--Yang formula for a two-dimensional Fermi gas: upper bound}

\author{Christian Hainzl, Fabian Saxler, Robert Seiringer}
\date{\footnotesize \today}

\renewcommand{\thefootnote}{\arabic{footnote}}
\setcounter{footnote}{0}

\maketitle
\abstract{We compute an upper bound on the ground state energy of a dilute two-dimensional Fermi gas with repulsive short-range interactions. Our bound can be viewed as the two-dimensional analogue of a formula derived by Huang and Yang in the three-dimensional case. It captures the first three terms in an asymptotic expansion for small $\varrho a^2$, where $\varrho$ denotes the density and $a$ the scattering length of the interaction potential.}

\tableofcontents
\newpage
\section{Introduction and main result}
Consider \(N\) identical fermions of spin \(\frac 12\) in a box \(\Lambda=[-L/2,L/2]^2\subset \bR^2\), with periodic boundary conditions for concreteness. The fermions interact via a non-negative potential $V$, which we write as the  periodization   \(V(x)=\sum_{z\in L\bZ^2}V_\infty(x+z)\) of a function  \(V_\infty\) that we assume to be nonnegative, radial and decaying sufficiently fast at infinity.
The relevant Hamiltonian  acts on  \(\bigwedge^NL^2(\Lambda,\bC^2)\), the Hilbert space of square integrable functions of variables \((x,\sigma)\in\Lambda\times\{\uparrow,\downarrow\}\) that are antisymmetric under permutations. In suitable units, it is given by
\begin{equation}\label{def:H}
	H_N=-\sum_{i=1}^{N}\Delta_{x_i}+\sum_{1\le i< j\le N} V(x_i-x_j)
\end{equation}
where \(\Delta_{x_i}\) is the Laplacian for the \(i\)'th (configuration space) coordinate. The Hamiltonian $H_N$ is bounded from below (in fact, non-negative) 
and we are interested in the ground state energy
\begin{equation}\label{eq:def enl}
	E(N_\uparrow,N_\downarrow,L)= \inf_{\psi\in\frh(N_\uparrow,N_\downarrow)}\frac{\langle\psi,H_N\psi\rangle}{\langle\psi,\psi\rangle}
\end{equation}
where \(\frh(N_\uparrow,N_\downarrow)\) denotes the subspace of \(\bigwedge^NL^2(\Lambda,\bC^2)\) with \(N_\uparrow\) particles of spin \(\uparrow\) and \(N_\downarrow=N-N_\uparrow\) particles of spin \(\downarrow\), which is left invariant by \(H_N\).
It follows from standard methods \cite{robinsonpressure} that the ground state energy density in the thermodynamic limit
\begin{equation}\label{eq:endensthl}
	e(\varrho_\uparrow,\varrho_\downarrow)=\lim_{\substack{L\rightarrow\infty\\
			N_\sigma/L^2=\varrho_\sigma,\ \sigma\in\{\uparrow,\downarrow\}
		}}     \frac{E(N_\uparrow,N_\downarrow,L)}{L^2}
\end{equation}
exists and is independent of boundary conditions. 

Our precise assumptions on $V_\infty$ are as follows. For the definition of the scattering length, we refer to \cite{LY2d}, \cite[Appendix C]{mathbosgas} or \Cref{sec:scatt eq} below.

\begin{A}\label[A]{ass:v-infty I}
	We assume \(V_\infty:\bR^2\rightarrow [0,\infty]\) to be measurable and radial, with scattering length \(a>0\).
	Moreover, there exists an \(R_0> 0\) such that
	\(	V_\infty(x)\le C a^{\varepsilon_0} /|x|^{2+\varepsilon_0}\)\ for all \(| x|\ge R_0\)
	for some \(\varepsilon_0>0\) and $C>0$.
\end{A}

In particular, we allow for an interaction between hard discs, corresponding to $V_\infty(x) = \infty$ for $|x| \leq a$, and zero otherwise.

To state our main result, we introduce the notation  \(\chi(\cdot)\) and \(\chi^c({\cdot})\) for the characteristic function of the (centered) unit ball and its complement, respectively.

\begin{T}\label[T]{thm:main-res}
	Let \(V_\infty\) satisfy  \Cref{ass:v-infty I}. There exists a $C>0$ such that the ground state energy density defined in \eqref{eq:endensthl} satisfies
	\begin{equation}\label{eq:main thm}
		e(\varrho_\uparrow,\varrho_\downarrow)\le 2\pi (\varrho_\uparrow^2+\varrho_\downarrow^2)+
		\frac{8\pi}{|\ln(\varrho a^2)|}\varrho_\uparrow\varrho_\downarrow+
		\frac{16\pi}{|\ln(\varrho a^2)|^2}F(\varrho_\downarrow/\varrho_\uparrow)\varrho_\uparrow\varrho_\downarrow
		+C \frac{\left( \ln |\ln a^2 \varrho| \right)^2}{|\ln(\varrho a^2)|^{3}} \varrho^2
	\end{equation}
	for small \(\varrho a^2=(\varrho_\uparrow+\varrho_\downarrow)a^2\), 
	where
	\begin{align}\label{eq:def F mein thm}
		F (t) &=  \gamma
		+\frac{1}{2}\ln\pi                                                     
		  -  \frac{(1+t)^2}{ \pi^3 t}
		\int_{\bR^{2\times 3}}
		\chi(k \sqrt{1+t})\chi(q \sqrt{1+1/t}) 
		\\ & \qquad \qquad \qquad \quad\times
		\left(\frac{
				\chi^c((k+p)\sqrt{1+t})\chi^c((q-p)\sqrt{1+1/t})}{ (k+p)^2-k^2+(q-p)^2-q^2}
		-\frac{\chi^c(p)  
		}{2|p|^2}\right)
		\dd k\dd q\dd p \nonumber
	\end{align}
	and \(\gamma\approx 0.577\) denotes the Euler--Mascheroni constant.
\end{T}

\begin{Rs}\label[R]{rem:main thm}
	\begin{enumerate}
		\item We expect that the upper bound captures that  energy density correctly to the third order, i.e., that \eqref{eq:main thm} is valid as a lower bound as  well (for a suitable choice of the constant $C$). 
		\item The constant \(C\) depends on \(V_\infty\) only through 
		the decay at infinity as in \Cref{ass:v-infty I}. In fact, 
		an inspection of the proof shows that the decay assumption can be relaxed to 
		      \begin{equation}
			      V_\infty( x)< C
			    |\ln(|x| a^{-1})|^{-23}|x|^{-2}  
		      \end{equation}
		      for large $|x|$ 
		      to obtain an error of the same magnitude as the one in \Cref{thm:main-res}. This condition is slightly stronger (by a power of a logarithm) than 
			 the condition necessary for the scattering length to be well-defined (which, for non-negative potentials, is the case  if and only if \(\int_{|x|>c}V_\infty(x)\ln(x/c)^2\dd x<\infty\) for some \(c>0\), see \cite{Landon2012}).
		\item   Up to the order \(\frac{\varrho^2}{|\ln(\varrho a^2)|}\)  the correct upper (and lower) bound was already obtained in \cite{lieb2005groundstateenergydilutefermi}. (For an extension to non-zero  
		 temperature see \cite{Seiringer_2005}.) 
		The novelty of the present bound is the higher precision, i.e., the next order term, which is the two-dimensional analogue of the Huang--Yang term \cite{HY} in three dimensions. 
		\item   In the fully spin-polarized case where either \(\varrho_\uparrow=0\) or \(\varrho_\downarrow=0\), the energy density reads \cite{Lauritsen_2024,lauritsen2024groundstateenergydilute}
		      \begin{equation}
			      e(\varrho,0)= 2\pi \varrho ^2+4\pi^2 a_p^2\varrho^3 +O(\varrho^4 |\ln(\varrho a_p^2)|^2)
		      \end{equation}
		      and the next-to-leading-order contribution is much smaller than  the new higher order contribution or the error term of \Cref{thm:main-res}. Here  \(a_p\) denotes the \(p\)-wave scattering length of $V_\infty$, 		      which is the relevant length scale for the interactions  between particles of equal spin.
		      (Results for non-zero temperatures are available in the fully polarized case as well \cite{Lauritsen_2024preslo,lauritsen2024pressuredilutespinpolarizedfermi}.)
		\item  For a two-dimensional Bose gas, where the Hilbert space consists of functions that are symmetric under permutations of the coordinates, the ground state energy density is given by \cite{fournais2022groundstateenergytwodimensional,LY2d}
		      \begin{equation}
			      e_\mathrm{Bos}(\varrho)= 4\pi \varrho^2 \delta \left(1-\delta|\log(\delta)|
			      + (2\gamma +\tfrac{1}{2} + \log(\pi))\delta
			      \right)
			      +o(\varrho^2\delta^{2}),
		      \end{equation}
		      where \(\delta=|\log(\varrho a^2)|^{-1}\), and the corrections differ from the fermionic energy density beginning from the first term. (See \cite{DEUCHERT_2020,Mayer_2020}    
		      for the corresponding problem at non-zero temperature, and \cite{Dyson,lieb2000groundstateenergydilute,FS1,FS2} for the three-dimensional case.)
		\item Our proof follows \cite{giacomelli2024huangyangformulalowdensityfermi} to some extent, where the corresponding problem in three dimensions is investigated. In fact, in \cite{giacomelli2024huangyangformulalowdensityfermi, giacomelli2025huangyangconjecturelowdensityfermi}  it was shown for the analogous three-dimensional problem that
		      \begin{equation}
			      e_\mathrm{3D}(\varrho_\uparrow,\varrho_\downarrow)= \frac{3}{5}(6\pi^2)^{\frac{2}{3}}(\varrho_\uparrow^{\frac{5}{3}}+\varrho_\downarrow^{\frac{5}{3}})+8\pi a\varrho_\uparrow\varrho_\downarrow+a^2\varrho_\uparrow^{\frac{7}{3}}F_\mathrm{3D}(\varrho_\downarrow/\varrho_\uparrow)+o(a^2\varrho^{7/3}) 
	            \end{equation}
		      for a suitable function \(F_\mathrm{3D}\).  Notably, \(F_\mathrm{3D}\) can be computed explicitly, but we are not aware of an explicit expression for the corresponding function in two dimensions given in  \eqref{eq:def F mein thm}. Compared to  \cite{giacomelli2024huangyangformulalowdensityfermi}, the two-dimensional case comes with additional difficulties, which we shall outline in the following.
	\end{enumerate}
\end{Rs}

\subsection{Strategy of the proof}

As in previous work on either the bosonic \cite{Basti_2026, Basti_2022, basti2021newsecondorderupper, Dyson, fournais2022groundstateenergytwodimensional, LY2d, Mayer_2020,YY} or fermionic case \cite{CWZ1,CWZ2,Falconi2021, giacomelli2024huangyangformulalowdensityfermi, Lauritsen_2024, lieb2005groundstateenergydilutefermi}, the basic strategy consists of constructing a trial state that has the correct correlation structure at small distances, taking the two-body interactions into account. Without it, one would obtain a wrong expression for the energy. In three dimensions, the effect is less severe, effectively the correlation structure replaces the appearance of $\int V$ in the formulas by $8 \pi a$ (with $a$ the three-dimensional scattering length), as already remarked in the famous footnote in Bogoliubov's original paper \cite{Bogolyubov1947OnTT}. In two dimensions, however, the effect is more pronounced, as $\int V$ gets replaced by 
 \(\frac{8\pi}{|\ln(\varrho a^2)|}\) which goes to zero in the dilute limit $a^2\varrho\to 0$. 

On the technical level, this means that if one tries to estimate the interaction by $\int V$ in some error terms, one effectively loses a factor $|\ln a^2\varrho|^{-1}$, which may render the estimate useless. Consequently, in the two-dimensional case one faces similar difficulties as one encounters  in three dimensions for interaction potentials that are not integrable.

To remedy this, we follow the method in \cite{Basti_2022,fournais2022groundstateenergytwodimensional} and construct the trial state via a two-step procedure. In the first step, one introduces a Jastrow factor that effectively leads to a problem where the original interaction $V_\infty$ has been replaced by a softer one $W_\infty$ 
with the same scattering length, but whose integral is of the right order $|\ln a^2\varrho|^{-1}$. In a second step, one constructs a trial state inspired by second order perturbation theory (for the soft potential $W_\infty$), applying an analysis similar to \cite{giacomelli2024huangyangformulalowdensityfermi}. 

Due to lack of a good control on the norm of the trial function because of the Jastrow factor, we need to localize particles to small boxes and cannot directly work in the thermodynamic limit, similarly as in \cite{Mayer_2020,fournais2022groundstateenergytwodimensional}. This is done in the first step in 
 \Cref{sec:2}. In the second step, we introduce the Jastrow factor in \Cref{sec:3}, with the estimate on its effect on the norm of the trial state postponed to \Cref{sec:errJ}. This effectively replaces the original interaction $V$ by a softer one $W$; the main bound on the ground energy of the resulting Hamiltonian will be given in \Cref{sec:setup}. Finally, we collect all estimates and complete the proof of \Cref{thm:main-res} in \Cref{sec:proof main res}.

For convenience, we shall in the following sections use the notation 
\begin{equation}\label{eq:def delta}
	\delta=|\ln(\varrho a^2)|^{-1}
\end{equation}
which we shall always assume to be small.

\section{Step 1. Reduction to smaller boxes} \label{sec:2}
A standard method \cite{robinsonpressure} to obtain an upper bound is to employ a trial state that considers particles localized in individual boxes. Small corridors are introduced to limit the interaction between particles in different boxes. Moreover, Dirichlet boundary conditions are needed to be able to glue together the state on small boxes to a state on the big box, but the effect of the change in boundary conditions can be quantified, as follows.

 Recall the definition \eqref{eq:def enl} for the ground state energy $E(N_\uparrow,N_\downarrow,L)$ of the Hamiltonian $H_N$ in \eqref{def:H}, defined with periodic boundary conditions on $[-L/2,L/2]^2$.  Following \cite[Sects.~II.C and II.D]{Mayer_2020} (and using the monotonicity in the particle number), one finds under  \Cref{ass:v-infty I} that
\begin{equation}\label{eq:step1}
e(\varrho_\uparrow,\varrho_\downarrow) \leq \ell^{-2} E(N_\uparrow, N_\downarrow,\ell - R_1 - 2R_2) + \frac {4\varrho}{R_2^2} + C \frac{\varrho^2 (1-R_1/\ell)^2}{(1-R_1/\ell - 2R_2/\ell)^4 } (a/R_1)^{\varepsilon_0}
\end{equation}
for any $R_2>0$, $R_1\geq R_0$  and $\ell > R_1 + 2 R_2$, and particle numbers $N_\uparrow\geq \varrho_\uparrow \ell^2$, $N_\downarrow\geq \varrho_\downarrow \ell^2$.  Here the localization error $4\varrho/R_2^2$ comes from switching between Dirichlet and periodic boundary conditions, and the last term bounds the interaction between different boxes that are separated by a distance $R_1$, taking our decay assumption on $V_\infty$ into account. 

The reduction in box size from $\ell$ to $\ell-R_1-2R_2$, and the resulting increase in density, leads to an error of order $\varrho^2 (R_1+R_2)\ell^{-1}$ in the energy density. In combination with the error term $\varrho R_2^{-2}$, the optimal choice of $R_2$ is thus $R_2 \sim (\ell/\varrho)^{\frac 13}$ and the resulting error is of oder $\varrho^2 (\varrho \ell^2)^{- \frac 13}$. Given  the last term in \eqref{eq:step1}, the optimal choice of $R_1$ is  $R_1 \sim (a^{\varepsilon_0} \ell)^{\frac 1{1+\varepsilon_0}}$, with a resulting error order $\varrho^2 (a/\ell)^{\frac {\varepsilon_0}{1+\varepsilon_0}}$. These error terms will thus be negligible if we choose $\ell \sim \varrho^{-\frac 12} \delta^{-\gamma_\ell}$ for some $\gamma_\ell \geq 9/2$.

\section{Step 2. Jastrow factor} \label{sec:3}
As already mentioned in the proof outline above, we cannot directly employ the method of \cite{giacomelli2024huangyangformulalowdensityfermi} for an upper bound on $E(N_\uparrow,N_\downarrow,L)$. Instead we first introduce a Jastrow factor in our trial function, which will effectively replace the interaction by a more regular and longer-ranged one, for which we can then construct a trial state as in \cite{giacomelli2024huangyangformulalowdensityfermi}. 

\subsection{The scattering equation in two dimensions}\label{sec:scatt eq}

We start by recalling  the definition of the scattering length and corresponding solution of the scattering equation. For details, we refer to \cite{LY2d} or \cite[Appendix~C]{mathbosgas}. 
We shall assume that $V_\infty$ is non-negative and radial. 

\begin{D}
\label[D]{def:scatt length}
	Let \(V_\infty:\bR^2\rightarrow [0,\infty]\) be measurable and radial. For $b>0$, define $a_b$ by  
	\begin{equation}\label{eq:scatt fucntional}
		\frac{4\pi}{\ln(b/a_b)}=\inf
		\left\{\int_{B_b}
		\left( 2|\nabla f |^2+V_\infty |f|^2\right) \ \Big|\ f\in H^1(B_b), \ f|_{\partial B_b}=1
		\right\}
	\end{equation}
	where $B_b$ denotes the (centered) ball of radius $b$. If $V_\infty$ is compactly supported, then $a_b$ is independent of $b$ for $b$ larger than the range of $V_\infty$, and equals its scattering length $a$. More generally, $a_b$ is increasing in $b$, and $a = \lim_{b \to \infty} a_b$ is finite  if and only if \(\int_{|x|>c}V_\infty(x)\ln(x/c)^2\dd x<\infty\) for some \(c>0\) \cite{Landon2012}. 
\end{D}

The minimizer \(f_{\infty,b}\) for \eqref{eq:scatt fucntional} is unique, and satisfies
\begin{equation}
	-2\Delta f_{\infty,b} +V_\infty f_{\infty,b} =0
\end{equation}
as a distribution on \(B_b\). 
We have  \(0\le f_{\infty,b}\le 1\), and \(f_{\infty,b}(x)= \frac{\ln(|x|/a)}{\ln(b/a)}\) outside a ball containing the support of \(V_\infty\). Moreover, \(f_{\infty,b}\) is radial and increasing in \(|x|\).

In the following, we shall consider
\(\varphi_{\infty,b}=1-f_{\infty,b}\), which satisfies
\(\varphi_{\infty,b}|_{\partial B_b}=0\).
We will extend \(\varphi_{\infty,b}\) by zero outside \(B_b\) and denote this function (now defined on all of \(\bR^2\)) by \(\varphi_{\infty,b}\) as well.
It satisfies 
\begin{equation}\label{eq:scatt eq whole space}
	2\Delta \varphi_{\infty,b} +V_\infty(1-\varphi_{\infty,b})=\frac{2}{R\ln(b/a_b)} \delta_{|x|=b} \quad  \mathrm{on}\ \bR^2.
\end{equation}
Importantly, we find via the divergence theorem that
\begin{equation}
	\int_{B_b} V_\infty (1-\varphi_{\infty,b})= \frac{4\pi}{\ln(b/a_b)}.
\end{equation}

\subsection{The Jastrow factor}

For $b>0$, consider the function 
 \(\varphi_{\infty,b}\) defined in the previous subsection, satisfying \eqref{eq:scatt eq whole space} (and being equal to zero outside \(B_b\)). Let $\varphi_b$ denote its periodization over the box \(\Lambda=[-L/2,L/2]^2\), and let $f_b = 1 - \varphi_b$.  
For any \(\Phi\in \bigwedge^NL^2(\Lambda,\bC^2)\), consider
the state \(\Psi\) defined by 
\begin{equation}\label{def:psi}
	\Psi(x_1,\ldots,x_N)= \prod_{1\le i<j\le N}f_b(x_i-x_j) \Phi(x_1,\ldots,x_n).
\end{equation}
For brevity, we shall denote $f=f_b$ and  \(f_{ij}=f_b(x_i-x_j)\) in the following.

For convenience, we shall denote by \(\cH_V\) the Hamiltonian \eqref{def:H} with interaction defined by $V_\infty$. Let 
\begin{equation}\label{def:Vtilde}
W_\infty(x) =\frac{2}{b\ln(b/a_b)} \delta_{|x|=b}
\end{equation}
which is the effective potential appearing on the right-hand side of \eqref{eq:scatt eq whole space}. 	The new potential \(W_\infty\) has a slightly decreased scattering length \( a_b\le a\) compared to the original interaction \(V_\infty\).
	By monotonicity, we may replace \(a_b\) by \(a\) for an upper bound, however. As for $V$ above, we shall write $W(x)=\sum_{z\in L\bZ^2}W_\infty(x+z)$ for the periodization of $W_\infty$ on the box $\Lambda$.

\begin{Le}
\label[Le]{lem:restrsoft pot}
	Let \(V_\infty\) satisfy \Cref{ass:v-infty I}, and let $b\geq R_0$. Then, with the definitions above, 
	\begin{equation}\label{lem1:eq}
		\langle \Psi,\cH_{ V}\Psi\rangle\le \langle \Phi,\cH_{W}\Phi\rangle- \langle\Phi, \cR_3\Phi\rangle + C\frac{ N(N-1)}{2 b^2} \left(\frac a b\right)^{\varepsilon_0} \| \Psi\|^2
	\end{equation}
	with 
	\begin{equation}\label{def:R3}
		\cR_3=\sum_{i,j,k}' \frac{\nabla f_{ij}}{f_{ij}} \cdot\frac{\nabla f_{ik}}{f_{ik}} \prod_{1\le l<m\le N}f_{lm}^2
	\end{equation}
	where the primed sum runs over all pairwise distinct indices \(i,j,k\) between \(1\) and \(N\).
\end{Le}
\begin{proof}
The last term in \eqref{lem1:eq} comes from estimating the size of $V_\infty$ outside a ball of radius $b$, using the decay assumption in 	 \Cref{ass:v-infty I}. The remaining part is a simple calculation using $|f|\leq 1$, see \cite[Lemma 5.2]{fournais2022groundstateenergytwodimensional}.
\end{proof}

The error term involving the three-particle potential $\mathcal{R}_3$ will be bounded in \Cref{lem:error small boxes jastrow} in \Cref{sec:errJ}.
The necessity to bound the norm of \(\Psi\) from below in terms of the norm of $\Phi$ forces us to choose the box \(\Lambda\) 
small enough, hence necessitating Step 1 above. In fact, we have the following bound.

\begin{Le}\label[Le]{lem:norm  jastrow}
With \(\Psi\) defined as in \eqref{def:psi}, 
	\begin{align}\label{def:R2}
		\|\Psi\|^2  
		            \ge \|\Phi\|^2 
		-\sum_{1\le i<j\le N} \int_{\Lambda^N} (1-f_{ij}^2)|\Phi|^2.
	\end{align}
\end{Le}

\begin{proof}
	This follows immediately from 
	\begin{equation}
				\prod_{1\le i<j\le N}f_{ij}^2\ge 1-\sum_{1\le i<j\le N}(1-f_{ij}^2).
	\end{equation}
\end{proof}

The error term can be viewed as an effective two-body interaction $\mathcal{R}_2$, which will be bounded in \Cref{prop:conj-q4-new} in the \Cref{sec:errJ}.

\section{Step 3. Upper bound for a soft interaction}
\label{sec:setup}
In this section 
we construct a trial state $\Phi$ to estimate   $\langle \Phi,\cH_{W}\Phi\rangle$ for the Hamiltonian with interaction potential $W$ given in \eqref{def:Vtilde}. The error terms $\langle\Phi, \cR_3\Phi\rangle$ and the one from \Cref{lem:norm jastrow} will be bounded in the next section. We start by introducing basic notions used throughout the proof.

\subsection{Notation}\label{sec:notation}
We denote by
\begin{equation}
	\hat{f}(p)= \int_\Lambda \dd x\, f(x)\ee^{-\ii p\cdot x},\qquad f(x)=\frac{1}{L^2}\sum_{\Lambda^*} \hat{f}(p)\ee^{\ii p\cdot x}
\end{equation}
the Fourier transform on the box \(\Lambda\),  
with  \(\Lambda^*= \frac{2\pi}{L} \bZ^2\). 
For integrals over \(\Lambda\) and sums over \(\Lambda^*\) we will frequently omit the explicit designation of the domain of integration or summation, respectively. 

For $\sigma \in \{ \uparrow,\downarrow\}$, let \(\hat v_\sigma,\hat u_\sigma: \Lambda^*\rightarrow \bR\),
\begin{equation}\label{eq:def hat v hat u}
	\hat v_\sigma(k)=
	\begin{cases}
		1\ \text{if}\ |k|\le k^\sigma_\text{F} \\
		0\ \text{if}\ |k|> k^\sigma_\text{F}
	\end{cases},
	\qquad \hat{u}_\sigma (k)=
	\begin{cases}
		0\ \text{if}\ |k|\le k^\sigma_\text{F} \\
		1\ \text{if}\ |k|> k^\sigma_\text{F}
	\end{cases}
	,
\end{equation}
denote the characteristic function of the Fermi ball \(B_\sigma\) and its complement, respectively. Here, \(B_\sigma=\{k\in \Lambda^*: |k|\le k_\mathrm{F}^\sigma\}\). The \(k_\mathrm{F}^\sigma\) will be suitably chosen to satisfy the particle number constraints $N_\sigma = |B_\sigma| \geq \varrho_\sigma \ell^2$ for the validity of \eqref{eq:step1}
(see also the discussion at the end of \Cref{sec:proof main res}).
We further define
\begin{equation}\label{eq:v-hat-t}
	\hat v_{t,\sigma} (p)= \ee^{t|p|^2} \hat v_\sigma(p),
	\qquad
	\hat u_{t,\sigma}(p)= \ee^{-t|p|^2} \hat u_\sigma (p).
\end{equation}

\subsection{Second quantization}\label{sec:2nd quant}
It will be convenient to work with in the formalism of second quantization. For the one-particle Hilbert space \(L^2(\Lambda,\bC^2)\), the corresponding Fock space \(\cF\)  is given by 
\begin{equation}
	\cF = \bigoplus_{n=0}^{\infty} \bigwedge^nL^2(\Lambda,\bC^2)
\end{equation}
where we identify \(\bigwedge^0 L^2(\Lambda,\bC^2)=\bC\). As is customary, the vacuum \((1,0,0,\ldots)\) will be denoted by \(\Omega\).

We denote by \(a^*(f),a(f)\)  the standard creation and annihilation operators satisfying  canonical anti-commutation relations
\begin{equation}
	\{a(f),a^*(g)\}= \langle f,g\rangle_{L^2(\Lambda,\bC^2)},\qquad
	\{a(f),a(g)\}=0
\end{equation}
for all \(f,g\) in \(L^2(\Lambda,\bC^2)\). They are bounded operators, satisfying
\begin{equation}\label{eq:norm ann}
	\|a(f)\|=\|f\|_{L^2(\Lambda,\bC^2)}.
\end{equation}

For \(\sigma\in\{\uparrow,\downarrow\},x\in\Lambda\) and \(k\in\Lambda^*=\frac{2\pi}{L}\bZ^2\), let \(a_{x,\sigma}\) denote the operator-valued distribution \( a(\delta_{\sigma,\cdot}\delta(x-\cdot))\). For
\(f_k(x)=\frac{1}{L}\ee^{\ii k\cdot x}\), we shall  write
\begin{equation}
	\hat a_{k,\sigma}= a(\delta_{\sigma,\cdot}f_k)= \frac{1}{L}\int_\Lambda \dd x\, a_{x,\sigma}\ee^{-\ii k\cdot x}.
\end{equation}
The number of particles with given spin $\sigma$ will be denoted by $\cN_\sigma$; it can be written as 
\begin{equation}
	\cN_\sigma=\sum_{\Lambda^*} \hat a^*_{k,\sigma}\hat a_{k,\sigma}
	=\int_\Lambda\dd x\, a^*_{x,\sigma} a_{x,\sigma}
	,\qquad \cN=\sum_{\sigma\in\{\uparrow,\downarrow\}} \cN_\sigma.
\end{equation}
With this notation, the second quantized extension \(\cH_V\) of \(H_N\) to \(\cF \) is 
\begin{equation}\label{eq:2nd quantized Hamiltonian}
	\cH= \sum_{\sigma\in\{\uparrow,\downarrow\}} \sum_{k\in\Lambda^*}|k|^2\hat a^*_k\hat a_k
	+\frac{1}{2L^2}
	\sum_{\sigma,\sigma'\in\{\uparrow,\downarrow\}} \sum_{k,p,q\in\Lambda^*}\hat V(k)\hat a^*_{p+k,\sigma}\hat a^*_{q-k,\sigma'} \hat a_{q,\sigma'}\hat a_{p,\sigma}.
\end{equation}

An important ingredient for the trial state is the particle-hole transformation \(R\) that plays the role of factoring out the energy of the free (non-interacting) Fermi gas. It acts on the creation/annihilation operators according to
\begin{equation}\label{eq:part hole}
	R^* a^*_{x,\sigma}R=a^*_\sigma(u_x)+ a_\sigma(v_x)
\end{equation}
where
\begin{equation}
	a^*_\sigma(u_x)=\int_\Lambda \dd y\, u_\sigma(x-y)a_{y,\sigma}^*,\qquad a_\sigma(v_x)=\int_\Lambda \dd y\, v_\sigma(x-y) a_{y,\sigma}
\end{equation}
and \(u_\sigma,v_\sigma,\) are the functions on \(\Lambda\) with Fourier coefficients defined in \eqref{eq:def hat v hat u}.
The ground state \(\psi_\mathrm{FFG}\) of the free Fermi gas of \(N=|B_\uparrow|+|B_\downarrow|\) particles is given by 
\begin{equation}\label{eq:ffg}
	\psi_\mathrm{FFG}= \prod_{\sigma\in\{\uparrow,\downarrow\}}\prod_{k\in B_\sigma} \hat{a}^*_{k,\sigma}\Omega = R\Omega.
\end{equation}

\subsection{Main result for soft potentials}\label{sec:upper bound for soft potential}

Our main result of this section is as follows. Recall the definition of $W$ in \eqref{def:Vtilde}, and the one of $\delta$ in \eqref{eq:def delta}.

\begin{T}\label[T]{thm:bound small box}
	Assume that \(L \geq \delta^{-3}\varrho^{-{\frac 12}}\) and \( b =\delta^{\gamma_b}\varrho^{-\frac 12}\) for some $\gamma_b \geq 1$. 
 For any \(\varrho_\uparrow\), \(\varrho_\downarrow\) such that
	\(N_\sigma=\varrho_\sigma L^2= |B_\sigma|\)  for \(\sigma\in\{\uparrow,\downarrow\}\), there is a normalized state \(\Phi\) in \( \bigwedge^N  L^2([-L/2,L/2]^2) \) satisfying
	$\cN_\sigma\Phi =N_\sigma \Phi$
	and
	\begin{align}
		\frac{\langle \Phi, \mathcal{H}_{W} \Phi\rangle}{L^2}
		\le  2\pi (\varrho_\uparrow^2+\varrho_\downarrow^2)+ 8\pi \delta \varrho_\uparrow\varrho_\downarrow + 
				16 \pi \delta^2 F(\varrho_\downarrow/\varrho_\uparrow)\varrho_\uparrow\varrho_\downarrow
		     + C \varrho^2 \delta^3 \left( \ln \delta^{-1} \right)^2 
			\label{eq:bound small box}
	\end{align}
	for \(\varrho a^2\) small enough.
	Here \(F\) is the function defined in \eqref{eq:def F mein thm}. 
\end{T}

The trial state $\Phi$ in \Cref{thm:bound small box} will be chosen as a unitary transformation \(R T\) of the Fock space vacuum \(\Omega\), where $R$ is the particle-hole transformation in \eqref{eq:part hole}, and \(T\) will be defined in the next subsection. Our task is thus to  compute  \(T^* R^* \mathcal{H}_{W} R T\) 
and then take its expectation value in the vacuum \(\Omega\).

Similarly as in \cite{giacomelli2024huangyangformulalowdensityfermi}, the unitary transformation $T$ 
can be viewed as a quasi-Bogoliu\-bov transformation 
implementing correlations among the particles,  taking into account the presence of the Fermi ball. 
It is actually substantially simpler than the one employed in  \cite{giacomelli2024huangyangformulalowdensityfermi}, where the correlations for both small and large length scales were implemented by quasi-bosonic Bogoliubov transformations, whereas the correlations on small length scales are already taken care of by the Jastrow factor of Step 2 in our approach.

Our method of proof is applicable to a more general class of soft potentials $W$ depending on a parameter $b$, not only the one introduced in \eqref{def:Vtilde}. In fact, what we shall need in the proof is that $W\geq 0$, 
\begin{equation}\label{int:W}
\int W \leq C \delta
\end{equation}
as well as 
\begin{equation}\label{need:W}
\frac 1{L^2} \sum_p \frac{|\hat W(p)|}{p^2 + \lambda}  \leq C \delta \ln (\lambda b^2)^{-1} 
\end{equation}
for small $\lambda b^2$ and $\lambda L^2  \gtrsim 1$.   
In our case,
\begin{equation}\label{eq:fourier transform of v}
		\hat W (p)=4\pi J_0(|p|b)/\ln(b/a)
 \end{equation} 
where $J_0$ is the zeroth spherical Bessel function, which is bounded and decays like $|p|^{-\frac 12}$ for large $p$, hence these conditions are all satisfied.

\subsection{Trial state}\label{sec:trial state}

Our trial state will be constructed with the aid of the following unitary transformation.

\begin{D}\label[D]{def:bosonic-transformations} For $\lambda\in \mathbb{R}$, let $T_{\lambda}  =  \ee^{\lambda(B - B^\ast)}$ with 
	\begin{align}\label{eq:def b}
		B =  \frac{1}{L^2}\sum_{p,r,r^\prime \in \Lambda^*}\hat{\eta}_{r,r^\prime}^\varepsilon(p) \hat{a}_{r+p,\uparrow}\hat{a}_{-r,\uparrow}
		\hat{a}_{r'-p,\downarrow}\hat{a}_{-r',\downarrow},
	\end{align}
	where 
	\begin{equation}\label{eq:def hat eta epsilon}
		\hat{\eta}^\varepsilon_{r,r^\prime}(p)=\hat {W}(p)
		\frac{\hat u_\uparrow (r+p) \hat v_\uparrow(r)\hat u_\downarrow (r'-p) \hat v_\downarrow(r')}{|r+p|^2 + |r'-p|^2 -|r|^2 - |r'|^2 +2\varepsilon} 
	\end{equation}
	and \(\varepsilon =\delta\varrho\).
	We shall write \(T\) for \(T_1\) for brevity.
\end{D}

The trial state is then
 given by 
\begin{equation}\label{def:Phi}
	\Phi=RT\Omega
\end{equation}
and fulfills \(\cN_\sigma \Phi=N_\sigma \Phi\) with \(N_\sigma=|B_\sigma|\), as can be seen from the fact that \(R\Phi\) has this property by construction, and that \(R^* B R\) commutes with \(\cN_\sigma\) for \(\sigma\in\{\uparrow,\downarrow\}\). 

\subsection{The correlation Hamiltonian}\label{sec:prel comp}

From the identity \(R\Omega = \psi_{\rm FFG}\) it follows that the energy of the ground state of the free Fermi gas equals  \(E_{\mathrm{FFG}}= \langle \psi_{\rm FFG},\cH_{W} \psi_{\rm FFG}\rangle=  \langle R\Omega, \mathcal{H}_{W} R\Omega\rangle\). We  compute \(R^*\cH_{W} R\) explicitly as in \cite{Falconi2021}, applying \eqref{eq:part hole} and normal ordering the resulting terms. The result is as follows.

\begin{Prop}
\label[Prop]{prop: fermionic transf}
	Let $\psi\in \mathcal{F}$ be a normalized state satisfying $\mathcal{N}_\sigma\psi = N_\sigma\psi$ for $\sigma\in\{\uparrow,\downarrow\}$. 
	Then
	\begin{equation}
		\langle \psi, \mathcal{H}_{W}\psi\rangle = E_{\mathrm{FFG}} + \langle R^\ast\psi,   \mathcal{H}_{\mathrm{corr}} R^\ast\psi\rangle
	\end{equation}
	where $E_{\mathrm{FFG}}$ is the energy of the free Fermi gas satisfying
	\begin{equation}\label{eq:energy ffg}
		\frac{E_\mathrm{FFG}}{L^2}= 2\pi(\varrho_\uparrow^2+\varrho_\downarrow^2)+\hat{W}(0)\varrho_\uparrow\varrho_\downarrow+O\left(\delta b^2\varrho^{3}\right)+O\left( {\varrho^{\frac{3}{2}}}{L^{-1}} \right)
	\end{equation}
	and   \(\mathcal{H}_{\mathrm{corr}}=~\mathbb{H}_0~+~\mathbb{X}~+~\sum_{i=1}^4\mathbb{Q}_i\) is the correlation Hamiltonian given by

	\begin{alignat}{2}
		 & \mathbb{H}_0 &                                                                                                                                                                                                                                                                         & = \sum_\sigma\sum_k ||k|^2 -(k_{\rm F}^\sigma)^2| \hat{a}_{k,\sigma}^\ast \hat{a}_{k,\sigma}, \nonumber                                                                                                                                                                                  \\
		 & \mathbb{X}   &                                                                                                                                                                                                                                                                         & = \sum_\sigma \int_{\Lambda^2}\dd x\dd y\, W(x-y)v_\sigma(x-y)\left(a^\ast_\sigma(u_x)a_\sigma(u_y) -a^\ast_\sigma(v_y)a_\sigma(v_x)\right)\nonumber   ,                                                                                                                                        \\
		 & \mathbb{Q}_1 &                                                                                                                                                                                                                                                                         & = \sum_{\sigma,\sigma^\prime}\int_{\Lambda^2} \dd x\dd y\, W(x-y) \Bigl[ a^\ast_\sigma(u_x)a^\ast_{\sigma}(v_x)a_{\sigma^\prime}(v_y)a_{\sigma^\prime}(u_y)\nonumber
		\\ & \	 &         & \qquad\qquad  + \tfrac{1}{2} a^\ast_\sigma(v_x)a^\ast_{\sigma^\prime}(v_y)a_{\sigma^\prime}(v_y)a_\sigma(v_x) - a^\ast_\sigma(u_x)a^\ast_{\sigma^\prime}(v_y)a_{\sigma^\prime}(v_y)a_\sigma(u_x)\Bigl]\nonumber,
		\\
		\label{eq: def H-corr}
		 & \mathbb{Q}_2 &                                                                                                                                                                                                                                                                         & = \frac{1}{2}\sum_{\sigma,\sigma^\prime} \int_{\Lambda^2} \dd x\dd y\, W(x-y) a^\ast_\sigma(u_x)a^\ast_{\sigma^\prime}(u_y)a^\ast_{\sigma^\prime}(v_y)a^\ast_{\sigma}(v_x) + \mathrm{h.c.}\nonumber,
		\\
		 & \mathbb{Q}_3 &          \nonumber                                                                                                                                                                                                                                                              & = \sum_{\sigma,\sigma^\prime} \int_{\Lambda^2} \dd x\dd y\, W(x-y) \Bigl[   a^\ast_\sigma(u_x) a^\ast_{\sigma^\prime}(v_y) a^\ast_{\sigma}(v_x)a_{\sigma^\prime}(v_y) \\  & \ & & \qquad \qquad - a^\ast_\sigma(u_x) a^\ast_{\sigma^\prime}(u_y) a^\ast_{\sigma}(v_x)a_{\sigma^\prime}(u_y) \Bigl] + \mathrm{h.c.}\nonumber,
		\\
		 & \mathbb{Q}_4 &                                                                                                                                                                                                                                                                         & = \frac{1}{2}\sum_{\sigma,\sigma^\prime}\int_{\Lambda^2} \dd x\dd y\, W(x-y)  {a}^\ast_\sigma(u_x){a}^\ast_{\sigma^\prime}(u_y){a}_{\sigma^\prime}(u_y){a}_\sigma(u_x) . 
	\end{alignat}
\end{Prop}

\Cref{eq:energy ffg} is obtained via a straightforward calculation \cite{Falconi2021} that yields 
\begin{equation}\label{4.23}
	\frac{E_\mathrm{FFG}}{L^2}= \frac{1}{L^2} \sum_\sigma \sum_{k\in B_\sigma} |k|^2   
		+ \frac{1}{2}\sum_{\sigma,\sigma'} 
	\left(\hat {W}(0)\varrho_\sigma\varrho_{\sigma'}-\frac{\delta_{\sigma,\sigma'}}{L^4}\sum_{k,k'\in B_\sigma} \hat {W}(k-k')\right).
\end{equation} 
 From the explicit form of \(W\) in \eqref{def:Vtilde} 
 one  finds that \(\hat {W}(p)=\hat {W}(0)+O(\delta b^2 |p|^2)\) and thus
 \begin{equation}\label{eq:sum fermi ball}
	\frac{1}{L^4}\sum_{k,k'\in B_\sigma} \hat {W}(k-k')=
	\hat {W}(0)\varrho_\sigma^2 +O(\delta b^2 \varrho^3).
 \end{equation}
The sum in the first term on the right-hand side of \eqref{eq:sum fermi ball} is easily seen to equal $2\pi (\varrho_\uparrow^2 + \varrho_\downarrow^2) + O(\varrho^{\frac 32} L^{-1})$, hence leading to \eqref{eq:energy ffg}.

We further decompose $\mathbb{Q}_2 =  \mathbb{Q}_2^\parallel  +  \mathbb{Q}_2^{\uparrow\downarrow}$ with  $\mathbb{Q}_2^\parallel$ including only same-spin interactions, i.e., 
\begin{alignat}{2}
	 & \mathbb{Q}_2^\parallel            &  & =  \frac{1}{2}\sum_{\sigma} \int_{\Lambda^2} \dd x\dd y\, W(x-y) a^\ast_\sigma(u_x)a^\ast_{\sigma}(u_y)a^\ast_{\sigma}(v_y)a^\ast_{\sigma}(v_x) + \mathrm{h.c.}\nonumber                                   \\
	 & \mathbb{Q}_2^{\uparrow\downarrow} &  & = \frac{1}{2}\sum_{\sigma\neq\sigma^\prime} \int_{\Lambda^2} \dd x\dd y\, W(x-y) a^\ast_\sigma(u_x)a^\ast_{\sigma^\prime}(u_y)a^\ast_{\sigma^\prime}(v_y)a^\ast_{\sigma}(v_x) + \mathrm{h.c.} \label{q2ud}
\end{alignat}
We shall see in the following that we can restrict our attention to the 
effective correlation Hamiltonian $\mathbb{H}_0 + \mathbb{Q}_2^{\uparrow\downarrow}$ 
with only a negligible error. In fact, 
as in \cite[Proposition 3.3]{Falconi2021}, 
\begin{equation}\label{eq:err simpl corr}
	|\langle\psi,\bX \psi\rangle|\le C \hat {W}(0)\varrho \langle\psi,\cN\psi\rangle,\qquad
	|\langle\psi,\bQ_1 \psi\rangle|\le C \hat {W}(0)\varrho \langle\psi,\cN\psi\rangle, 
\end{equation}
and $\hat {W}(0) \leq C \delta$ in our case. Moreover, as in \cite{giacomelli2024huangyangformulalowdensityfermi}, \(\bQ_3\) and \(\bQ_2^\parallel\) vanish on the considered trial state:
\begin{equation}\label{q320}
\langle T \Omega , \bQ_3 T \Omega \rangle = \langle T \Omega , \bQ_2^\parallel T \Omega \rangle = 0\,.
\end{equation}
To see this note that  \(B\) creates or annihilates four particles, two with spin up and two with spin down. The operator
 \(\bQ_3\) creates/annihilates two particles, while \(\bQ_2^\parallel\) creates/annihilates four 
  particles of the same spin, hence \(\bQ_3T\Omega\) and \(\bQ_2^\parallel T\Omega\) are orthogonal to \(T\Omega\).

In \Cref{sec:q4} it will be shown via a more involved computation that the contribution of \(\bQ_4\) is small as well. In contrast, for the three-dimensional case considered in \cite{giacomelli2024huangyangformulalowdensityfermi}, \(\bQ_4\) is essential for the renormalization of the interaction at high momenta and negligible only for small momenta. Of course, the soft potential considered here should be thought of as already renormalized.
In our case, the operator  \(B\) is determined by  solving the commutator equation 
\begin{equation}\label{eq:commutator h0 q2}
	[\bH_0,B-B^*]\approx -\bQ_2^{\uparrow\downarrow},
\end{equation}
which corresponds to the second (low momentum) quasi-bosonic Bogoliubov transformation in \cite{giacomelli2024huangyangformulalowdensityfermi}, or the part of the Bogoliubov transformation acting on low momenta in \cite{Basti_2026}, respectively.

\subsection{Configuration space representation and $t$-integral bounds}
Following \cite{giacomelli2024huangyangformulalowdensityfermi}, it will be very useful to write
\begin{equation}\label{def:inte}
	\frac{1}{|r+p|^2 + |r'-p|^2 -|r|^2 - |r'|^2 +2\varepsilon}= \int_0^\infty \dd t\, \ee^{-2t\varepsilon}
	\ee^{-t|r+p|^2}\ee^{t|r|^2}\ee^{-t|r'-p|^2}\ee^{t|r'|^2}
\end{equation}
on the support of \(\hat u_\uparrow(p+r)\hat u_\downarrow(r'-p)\hat v_\uparrow(r)\hat v_\downarrow(r')\).
By applying this in \eqref{eq:def hat eta epsilon}, we obtain the configuration space representation 
\begin{align}\label{eq:b pos space}
	B & = \int_0^\infty\dd t\, \ee^{-2t\varepsilon} \int\dd x \dd y\,
	W(x-y) a_\uparrow(u_{t,x})a_\uparrow(v_{t,x})a_\downarrow(u_{t,y})a_\downarrow(v_{t,y})
\end{align}
for the operator \(B\) defined in \eqref{eq:def b}, with $u_{t,\sigma}$ and $v_{t,\sigma}$ defined in \eqref{eq:v-hat-t}. 
Using this representation, there often arises a need for bounds of the following type throughout the proof.

\begin{Le}
\label[Le]{lem:integration-t-new}
	For $0<\varepsilon<(k_{\rm F}^\sigma)^2$ and $L > k_{\rm F}^\sigma/\varepsilon$, 
	\begin{equation}\label{eq:lem:integration-t-1-new}
		  \int_0^\infty \dd t\, \ee^{-2\varepsilon t} \|u_{t,\sigma}\|_2 \|v_{t,\sigma}\|_2\le C  \ln \left( 1+ \tfrac{(k_{\rm F}^\sigma)^2}\varepsilon\right)  .
	\end{equation}
	Moreover, for $\lambda \geq 1$, 
	\begin{equation}\label{eq:lem:integration-t-1-new2}
		  \int_{(\lambda k_{\rm F}^\sigma)^{-2}}^\infty \dd t\, \ee^{2 t ( (k_{\rm F}^\sigma)^2 - \varepsilon)} \|u_{t,\sigma}\|_2^2  \le C\lambda  \ln \left( 1+ \tfrac{(k_{\rm F}^\sigma)^2}\varepsilon\right)  .
	\end{equation}
\end{Le}

\begin{proof}
Let $f(t) = t^{\frac 12} \ee^{2t}/ ( 1+ t^{\frac 12})$ for $t>0$. By Cauchy--Schwarz, 
	\begin{align}\nonumber
		  \int_0^\infty \dd t\, \ee^{-2\varepsilon t} \|u_{t,\sigma}\|_2 \|v_{t,\sigma}\|_2 & \le \left(  \int_0^\infty \dd t\, \ee^{-2\varepsilon t} \|u_{t,\sigma}\|_2^2 f( t (k_{\rm F}^\sigma)^2) 
		  \right)^{\frac 12} \\ & \quad  \times \left(  \int_0^\infty \dd t\, \ee^{-2\varepsilon t} \|v_{t,\sigma}\|_2^2 f( t (k_{\rm F}^\sigma)^2)^{-1} 
		  \right)^{\frac 12}\label{eq:lem:integration-t-1-newx}.
	\end{align}
For the first term on right-hand side, we can bound
\begin{equation}
  \int_0^\infty \dd t\, \ee^{-2(1+x) t}  f( t ) \leq \frac {C}{x\sqrt{1+x}} 
\end{equation}
for $x>0$, which yields
\begin{equation}\label{4:33}
 \int_0^\infty \dd t\, \ee^{-2\varepsilon t} \|u_{t,\sigma}\|_2^2 f( t (k_{\rm F}^\sigma)^2) \leq \frac{C}{L^2} \sum_{|k|> k_{\rm F}^\sigma} \frac{1}{ |k|^2-(k_\mathrm{F}^\sigma)^2+\varepsilon} 
 \frac  {k_{\rm F}^\sigma}{|k|}
 \leq  C  \ln \left( 1+ \tfrac{(k_{\rm F}^\sigma)^2}\varepsilon\right)  .
 \end{equation}
 The last estimate is obtained by bounding the sum via the corresponding integral, which can easily justified under the stated assumption $L > k_{\rm F}^\sigma/\varepsilon$. 
The bound \eqref{eq:lem:integration-t-1-new2} follows immediately from \eqref{4:33} since $f(t) \geq \ee^{2t} / (2\lambda) $ for $t\geq 1/\lambda^2$ and $\lambda \geq 1$.

For the second term in \eqref{eq:lem:integration-t-1-newx}, on the other hand, we have
\begin{align}\nonumber
 \int_0^\infty \dd t\, \ee^{-2\varepsilon t} \|v_{t,\sigma}\|_2^2 f( t (k_{\rm F}^\sigma)^2)^{-1}  &  =	\frac{1}{2L^2} \sum_{|k| \leq k_{\rm F}^\sigma } \frac{1 + \sqrt{\pi} \sqrt{ 1 - (|k|^2-\varepsilon)/(k_\mathrm{F}^\sigma)^2}}{(k_\mathrm{F}^\sigma)^2-|k|^2+\varepsilon} \\ & \leq C  \ln \left( 1+ \tfrac{(k_{\rm F}^\sigma)^2}\varepsilon\right)  ,
\end{align}
hence proving \eqref{eq:lem:integration-t-1-new}.
\end{proof}

For our choice of $\varepsilon = \delta\varrho$, we have $\ln( 1+ \tfrac{(k_{\rm F}^\sigma)^2}\varepsilon) \leq C \ln \delta^{-1}$.   

\subsection{Conjugation of $\cN$}

We start by  computing bounds for \(T^*\cN T\).

\begin{Prop}
\label[Prop]{prop:bound-n}
	For $\lambda\in [0,1]$  
	\begin{equation}\label{eq:conj n t2eg prop}
		\langle T_\lambda \Omega,\cN T_\lambda\Omega\rangle
		\le C L^2\varrho \left( \delta \ln \delta^{-1}\right)^2 \,.
	\end{equation}
\end{Prop}

\begin{proof}
	We have $\langle\Omega,\cN\Omega\rangle=0$ and 
	\begin{equation}\label{tbug}
			\partial_\lambda\langle \psi  ,
			 T_{\lambda}^*  \cN T_{\lambda} \psi\rangle
			= \langle\psi, T_{\lambda}^*[\cN,B-B^*] T_{\lambda}\psi\rangle
			  = -4\langle \psi,  T_{\lambda}^* (B+B^*) T_{\lambda} \psi \rangle
	\end{equation}
for any normalized \(\psi\in\cF\).
	We will apply the representation \eqref{eq:b pos space} for $B$, 
	and proceed as in \cite[Prop.~3.5]{giacomelli2024huangyangformulalowdensityfermi}. 
		We  have 
		\begin{align}
		|\langle \psi, B \psi\rangle| & \le
		\int\dd t\,\ee^{-2t\varepsilon} \|u_{t,\uparrow}\|_2\|v_{t,\uparrow}\|_2\|v_{t,\downarrow}\|_2 \int \dd y
		\, \|a_\downarrow(u_{t,y}) \psi\|    \int W \,  .
	\end{align}
	Since $\|v_{t,\downarrow}\|_2 \leq \ee^{t (k_{\rm F}^\downarrow)^2} \|v_\downarrow\|_2$ and $0\leq \hat u_{t,\downarrow}(k) \ee^{t (k_{\rm F}^\downarrow)^2} \leq 1$, 
	the Cauchy--Schwarz inequality yields
	\begin{align}
		\|v_{t,\downarrow}\|_2\int \dd y
		\, \|a_\downarrow(u_{t,y}) \psi\| 
		 & \le   \|v_\downarrow\|_2 L \left( \int \dd y
		\, \|a_\downarrow(u_{t,y}) \psi\|^2 \ee^{2 t (k_{\rm F}^\downarrow)^2} \right)^{\frac 12} 
		 \le 
		 \|v_\downarrow\|_2 L  \|\cN^{\frac{1}{2}}\psi\|  \label{eq:n i small} .
	\end{align}
 	Using further that 
	$\int W \leq C \delta$ for our choice of parameters, 
	  \Cref{lem:integration-t-new} thus implies that 
	\begin{equation}\label{bound:b}
		|\langle \psi, B \psi\rangle| \le  C L\varrho^{\frac{1}{2}}\|\cN^{\frac{1}{2}}\psi\| \delta \ln \delta^{-1}  \,.
	\end{equation}
	\Cref{eq:conj n t2eg prop} then follows from \eqref{tbug} via Gr\"onwall's Lemma.
\end{proof}

\subsection{Conjugation of \(\bH_0\)}\label{sec: h0}
We shall now show that, up to a negligible error, the commutator of \(\bH_0\) with \(B-B^*\) satisfies \eqref{eq:commutator h0 q2}.
\begin{Prop}
\label[Prop]{prop:conj-h-0}
	We have  
	\begin{equation}
		\partial_\lambda T_{\lambda}^* \bH_0 T_{\lambda}= T_{\lambda}^* (-\bQ_{2}^{\uparrow\downarrow}+\cE_{\bH_0})T_{\lambda}
	\end{equation}
	with \(\cE_{\bH_0}=  2 \varepsilon (B + B^*)\) satisfying 
	\begin{equation}\label{4.42}
		|\langle \psi,\cE_{\bH_0}\psi\rangle|\le C L\varrho^{\frac{3}{2}} \|\cN^{\frac{1}{2}}\psi\| \delta^2 \ln \delta^{-1}
	\end{equation} 
	 for any normalized \(\psi\in\cF\).
\end{Prop}
\begin{proof}
	A simple  computation shows that  
	\begin{equation}
		\partial_\lambda T_\lambda^* \bH_0T_\lambda= T_\lambda^*[\bH_0,B-B^*]T_\lambda=T_\lambda^* \left( -\bQ_2^{\uparrow\downarrow}+2\varepsilon(B+B^*) \right) T_\lambda \,.
	\end{equation}
	The expectation value of \(B+B^*\) has been estimated in \eqref{bound:b} already. With $\varepsilon = \delta\varrho$, this yields \eqref{4.42}. 
\end{proof}

\subsection{Conjugation of \(\bQ_2^{\uparrow\downarrow}\)}\label{sec:q2}
In this subsection, we compute the contribution due to \(\bQ_2^{\uparrow\downarrow}\), and prove that it is given by the correct constant up to sub-leading terms.
\begin{Prop}
\label[Prop]{prop:conj-bq2}
	We have 
	\begin{equation}\label{eq:prop conj q2}
		\partial_\lambda T_{\lambda}^* \bQ_2^{\uparrow\downarrow}T_{\lambda}= 
		-\frac{2}{L^4}\sum_{p,r,r'\in\Lambda^*}\hat W(p)
		 \hat \eta^\varepsilon_{r,r'}(p) \hat u_\uparrow(r+p)\hat u_\downarrow(r'-p)\hat v_\uparrow(r)\hat v_\downarrow(r')
		 +T_{\lambda}^*\cE_{\bQ_2}T_{\lambda}
	\end{equation}
	with $\cE_{\bQ_2}$ satisfying 
	\begin{equation}\label{eq:bound ce q-2less t2}
		|\langle\psi,\cE_{\bQ_2}\psi\rangle|\le C \varrho \langle\psi,\cN\psi\rangle\delta^{2} \ln \delta^{-1} 
		+	C \varrho^{\frac{1}{2}} \|\bQ_4^{\frac{1}{2}}\psi\| \|\cN^{\frac{1}{2}}\psi\|\delta^{\frac{3}{2}}\ln\delta^{-1}
	\end{equation}
	for any \(\psi\in\cF\).
\end{Prop}

\begin{proof}
	The derivative is given by  
	\begin{equation}
		\partial_\lambda T_{\lambda}^* \bQ_2^{\uparrow\downarrow}T_{\lambda} 
		=T_\lambda^*[\bQ_2^{\uparrow\downarrow},B]T_\lambda+\mathrm{h.c.}
	\end{equation}
	Recalling the definitions of \(\bQ_2^{\uparrow\downarrow}\) and \(B\) in in \eqref{q2ud}  and \eqref{eq:def b}, respectively, we see that the relevant commutator to be calculated is
	\begin{equation}
          [\hat{a}_{-s^\prime, \uparrow}^\ast\hat{a}_{s^\prime - k,\uparrow}^\ast\hat{a}_{-s,\downarrow}^\ast\hat{a}_{s+k,\downarrow}^\ast, \hat{a}_{r+p,\uparrow}\hat{a}_{-r,\uparrow}\hat{a}_{r^\prime - p, \downarrow}\hat{a}_{-r^\prime, \downarrow}]
          \end{equation}
 for $r',s \in B_\downarrow$, $r,s' \in B_\uparrow$, $s+k,r'-p \not\in B_\downarrow$ and $s'-k,r+p \not \in B_\uparrow$. Under these conditions, the commutator is given, in normal order, by a sum of $15$ terms, 
	 \begin{align}\nonumber
          &\delta_{s^\prime, r} \hat{a}^\ast_{s^\prime-k, \uparrow}  \hat{a}^\ast_{-s,\downarrow} \hat{a}_{s+k,\downarrow}^\ast  \hat{a}_{r+p,\uparrow}\hat{a}_{r^\prime -p,\downarrow}\hat{a}_{-r^\prime, \downarrow}
          +\delta_{s,r^\prime} \hat{a}_{-s^\prime, \uparrow}^\ast \hat{a}^\ast_{s^\prime-k, \uparrow} \hat{a}_{s+k,\downarrow}^\ast \hat{a}_{r+p,\uparrow} \hat{a}_{-r,\uparrow}  \hat{a}_{r^\prime -p,\downarrow}
          \\
         &  +  \delta_{s^\prime - k, r+p} \hat{a}_{-s^\prime, \uparrow}^\ast  \hat{a}^\ast_{-s,\downarrow} \hat{a}^\ast_{s+k,\downarrow}
         \hat{a}_{-r,\uparrow}  \hat{a}_{r^\prime - p,\downarrow}    \hat{a}_{-r^\prime, \downarrow}
        +   \delta_{s+k,r^\prime -p}           \hat{a}_{-s^\prime, \uparrow}^\ast \hat{a}^\ast_{s^\prime - k, \uparrow}
         \hat{a}^\ast_{-s,\downarrow} \hat{a}_{r+p,\uparrow}  \hat{a}_{-r,\uparrow}\hat{a}_{-r^\prime, \downarrow}
        \nonumber     
          \\ &  - \delta_{s^\prime, r} \delta_{s,r^\prime}   
         \hat{a}^\ast_{s^\prime-k, \uparrow}\hat{a}_{s+k,\downarrow}^\ast \hat{a}_{r+p,\uparrow}\hat{a}_{r^\prime -p,\downarrow}
                   -   \delta_{s^\prime -k, r+p}\delta_{s+k,r^\prime - p}
        \hat{a}_{-s^\prime, \uparrow}^\ast  
         \hat{a}^\ast_{-s,\downarrow} \hat{a}_{-r,\uparrow} \hat{a}_{-r^\prime, \downarrow}
          \nonumber         
          \\
         &   - \delta_{s,r'}  \delta_{s^\prime - k, r+p}  \hat{a}_{-s^\prime, \uparrow}^\ast  \hat{a}^\ast_{s+k,\downarrow} 
        \hat{a}_{-r,\uparrow} \hat{a}_{r^\prime - p,\downarrow}  
           - \delta_{s',r}   \delta_{s+k,r^\prime -p}\hat{a}^\ast_{s^\prime - k, \uparrow}
            \hat{a}^\ast_{-s,\downarrow}\hat{a}_{r+p,\uparrow} \hat{a}_{-r^\prime, \downarrow}
        \nonumber
          \\
         & -  \delta_{s,r'} \delta_{-k,p} \hat{a}_{-s^\prime, \uparrow}^\ast \hat{a}^\ast_{s^\prime - k, \uparrow}
        \hat{a}_{r+p,\uparrow} \hat{a}_{-r,\uparrow} -   \delta_{s',r} \delta_{- k, p} \hat{a}^\ast_{-s,\downarrow} \hat{a}^\ast_{s+k,\downarrow} 
          \hat{a}_{r^\prime - p,\downarrow} 
           \hat{a}_{-r^\prime, \downarrow}
                \nonumber
\\
         &   +  \delta_{s,r'}\delta_{s',r}  \delta_{ - k, p} \left( \hat{a}^\ast_{s+k,\downarrow} \hat{a}_{s+k,\downarrow}  
               +  \hat{a}^\ast_{r+p, \uparrow}\hat{a}_{r+p,\uparrow}
          +         \hat{a}^\ast_{-s,\downarrow} \hat{a}_{-s, \downarrow}
          +         \hat{a}_{-r, \uparrow}^\ast  \hat{a}_{-r,\uparrow} - 1 \right) 
            \end{align}
	 and we will split \([\bQ_2^{\uparrow\downarrow},B]=\sum_{i=1}^{15}\mathrm{I}_i\) accordingly. The constant term in \eqref{eq:prop conj q2} is the last one,  with the remaining terms constituting \(\cE_{\bQ_2}\).

	The following bounds are analogous to similar ones in the three-dimensional case in  \cite[Props.~4.6 \& 5.5]{giacomelli2024huangyangformulalowdensityfermi}. 
	For a given term \(\mathrm I_j\), using \eqref{eq:b pos space}, we will denote by \(\mathrm I^t_j \) an expression such that \(\mathrm I_j =\int_0^\infty\dd t\, \ee^{-2t\varepsilon}\, \mathrm I^t_j \). The first term to consider is 
	\begin{align}\nonumber
	\mathrm I_1^t & = \frac 1{L^2} \sum_r \hat v_{t,\uparrow}(r) \int \dd x \dd y \dd z \dd z' \, \ee^{\ii r\cdot(x-z)} W(x-y) W(z-z')  \\ & \qquad\qquad \times a_\uparrow^*(u_x) a_\downarrow^*(u_y) a_\downarrow^*(v_y) a_\downarrow(u_{t,z'}) a_\downarrow(v_{t,z'}) a_\uparrow(u_{t,z}) .
	\end{align}
	Using that $0\leq \hat v_{t,\uparrow}(r) \leq \ee^{t (k_{\rm F}^\uparrow)^2}$ the Cauchy--Schwarz inequality yields the bound
	\begin{align}\nonumber
	|\langle\psi,\mathrm{I}^t_1\psi\rangle|&  \le \ \ee^{t (k_{\rm F}^\uparrow)^2}  \left( \int \dd x \left\| \int  \dd y \, W(x-y)   a_\downarrow(v_y) a_\downarrow(u_y)  a_\uparrow(u_x)\psi  \right\|^2 \right)^{\frac 12} \\&  \quad \times  \left( \int \dd z   \left\|   \int  \dd z' \, W(z-z')  a_\downarrow(u_{t,z'}) a_\downarrow(v_{t,z'}) a_\uparrow(u_{t,z}) \psi  \right\|^2 \right)^{\frac 12}  .
	\end{align}
	The first factor on the right-hand side can be bounded, again via Cauchy--Schwarz, as 
	\begin{align}
	 \int \dd x \left\| \int  \dd y \, W(x-y)   a_\downarrow(v_y) a_\downarrow(u_y)  a_\uparrow(u_x)\psi  \right\|^2 \leq \| v_\downarrow\|_2^2  \langle \psi , \bQ_4 \psi \rangle \int W  .
	 \end{align}
	In the second factor we can bound  
	\begin{align}\label{use:I1}
	  \left\|   \int  \dd z' \, W(z-z')  a_\downarrow(u_{t,z'}) a_\downarrow(v_{t,z'}) a_\uparrow(u_{t,z}) \psi \right\| \leq  \|u_{t,\downarrow}\|_2 \| v_{t,\downarrow}\|_2  \left\|   a_\uparrow(u_{t,z}) \psi \right\| \int W .
	\end{align}
	In combination with the second inequality in \eqref{eq:n i small}, this shows that 
	\begin{equation}
	|\langle\psi,\mathrm{I}^t_1\psi\rangle| \leq \|u_{t,\downarrow}\|_2 \| v_{t,\downarrow}\|_2 \| v_\downarrow\|_2 \|\bQ_4^{\frac{1}{2}}\psi\|  \|\cN^{\frac{1}{2}}\psi\| \left( \int W\right)^{3/2}.
	\end{equation}
	With the aid of \Cref{lem:integration-t-new} and \eqref{int:W}, this finally yields
	\begin{equation}\label{I1_final}
	|\langle\psi,\mathrm{I}_1\psi\rangle| \leq C \varrho^{\frac 12}  \|\bQ_4^{\frac{1}{2}}\psi\|  \|\cN^{\frac{1}{2}}\psi\| \delta^{\frac 3 2}\ln \delta^{-1} \,. 
	\end{equation}
	The term $\mathrm I_2$ can be bounded in the same way, simply exchanging $\uparrow$ and $\downarrow$.
	
	For $\mathrm I_3$, we similarly have
	\begin{align}\nonumber
	\mathrm I_3^t & = \frac 1{L^2} \sum_r \hat u_{t,\uparrow}(r) \int \dd x \dd y \dd z \dd z' \, \ee^{\ii r\cdot(x-z)} W(x-y) W(z-z') \\ & \qquad \qquad \ \times  a_\uparrow^*(v_x) a_\downarrow^*(u_y) a_\downarrow^*(v_y) a_\downarrow(u_{t,z'}) a_\downarrow(v_{t,z'}) a_\uparrow(v_{t,z}) .
	\end{align}
	Using that $0\leq \hat u_{t,\uparrow}(r) \leq \ee^{-t (k_{\rm F}^\uparrow)^2}$ and the Cauchy--Schwarz inequality, its expectation value can be bounded by 
	\begin{align}\nonumber
	|\langle\psi,\mathrm{I}^t_3\psi\rangle| &\le \ee^{-t (k_{\rm F}^\uparrow)^2}  \left( \int \dd x \left\| \int  \dd y \, W(x-y)  a_\downarrow(v_y) a_\downarrow(u_y) a_\uparrow(v_x) \psi \right\|^2 \right)^{\frac 12}\\ & \quad \times  \left( \int \dd z   \left\|   \int  \dd z' \, W(z-z')  a_\downarrow(u_{t,z'}) a_\downarrow(v_{t,z'}) a_\uparrow(v_{t,z}) \psi \right\|^2 \right)^{\frac 12}  .
	\end{align}
	The first factor can be bounded, again via Cauchy--Schwarz, as 
	\begin{equation}
	 \int \dd x \left\| \int  \dd y \, W(x-y)  a_\uparrow(v_x) a_\downarrow(u_y) a_\downarrow(v_y) \psi \right\|^2 \leq \left ( \int W \right)^2 \| v_\uparrow\|_2^2 \| v_\downarrow\|_2^2 \langle \psi, \cN \psi\rangle .
	 \end{equation}
	 For the second we can proceed as for $\mathrm I_1$ above, using \eqref{use:I1} with $v_{t,z}$ in place of $u_{t,z}$, as well as 
	 \begin{equation}\label{4.59}
	 \ee^{-2t (k_{\rm F}^\uparrow)^2}  \int dz  \left\|   a_\uparrow(v_{t,z}) \psi \right\|^2 \leq \langle \psi , \cN \psi\rangle.
	 \end{equation}
	Applying again \Cref{lem:integration-t-new}, we conclude that
	\begin{equation}\label{I3_final}
	|\langle\psi,\mathrm{I}_3\psi\rangle| \leq C \varrho \langle \psi, \cN \psi \rangle \delta^{2}\ln \delta^{-1}\,.  
	\end{equation}
	The term $\mathrm I_4$ is bounded in the same way,  exchanging $\uparrow$ and $\downarrow$.

	We proceed with $\mathrm I_5$, for which we have
	\begin{equation}
		\mathrm{I}^{t}_{5} \!=\! \int \dd x\dd y\dd z \dd z'\, W(x-y) W(z-z') v_{t,\downarrow}(y-z')v_{t,\uparrow}(x-z) a^*_\uparrow (u_x)a^*_\downarrow(u_y) a_\downarrow(u_{t,z'})a_\uparrow(u_{t,z}).
	\end{equation}
	Again we can use the  Cauchy--Schwarz inequality to bound it as  
	\begin{equation}\label{eq:i 5 t less}
		|\langle\psi,\mathrm{I}^t_{5}\psi\rangle|\le C \delta^{\frac{3}{2}} \varrho^{\frac{1}{2}}\|\bQ_4^{\frac{1}{2}}\psi\| \|\cN^{\frac{1}{2}}\psi\| \|u_{t,\uparrow}\|_2\|v_{t,\uparrow}\|_2,
	\end{equation}
	and hence \Cref{lem:integration-t-new} shows that also this term satisfies the bound \eqref{I1_final}. 
	
	The term $\mathrm I_6$ has a similar structure, but with interchanged roles of $u$ and $v$, i.e.,
	\begin{equation}
		\mathrm{I}^{t}_{6} \!=\!  \int \dd x\dd y\dd z \dd z'\, W(x-y) W(z-z') u_{t,\downarrow}(y-z')u_{t,\uparrow}(x-z)   a^*_\uparrow (v_x)a^*_\downarrow(v_y) a_\downarrow(v_{t,z'})a_\uparrow(v_{t,z}).
	\end{equation}
	We shall utilize two different bounds, depending on the magnitude of $t$. On the one hand, a simple Cauchy--Schwarz inequality shows that 
	\begin{equation}\label{eq:i 5 t less2}
		|\langle\psi,\mathrm{I}^t_{6}\psi\rangle|\le  C \min\{\varrho_\uparrow,\varrho_\downarrow\} \left( \int W\right)^2  \|u_{t,\uparrow}\|_2  \|u_{t,\downarrow}\|_2  
		\ee^{t (k_{\rm F}^\uparrow)^2+t (k_{\rm F}^\downarrow)^2} \langle \psi, \cN \psi\rangle .
	\end{equation}
	Let $k_{\rm F} = \max\{k_{\rm F}^\uparrow, k_{\rm F}^\downarrow\}$. We shall integrate \eqref{eq:i 5 t less2}  over $t \geq k_{\rm F}^{-2}$ using \Cref{lem:integration-t-new}. In fact, with a Cauchy--Schwarz inequality the bound \eqref{eq:lem:integration-t-1-new2} implies that 
	\begin{equation}
	\int_{k_{\rm F}^{-2}}^\infty \dd t \,  \|u_{t,\uparrow}\|_2  \|u_{t,\downarrow}\|_2 \ee^{t (k_{\rm F}^\uparrow)^2+t (k_{\rm F}^\downarrow)^2} \leq C \left( \frac{k_{\rm F}} { \min\{ k_{\rm F}^\uparrow,k_{\rm F}^\downarrow\}}\right)^{\frac 12}  \ln \delta^{-1} 
	\end{equation}
	which, after multiplication by $\min\{\varrho_\uparrow,\varrho_\downarrow\}$, is bounded by $\varrho \ln \delta^{-1}$. 
	On the other hand, 
	writing 
	\begin{align}\nonumber
		\mathrm{I}^{t}_{6}  = \frac 1{L^4} \sum_{r,r'} \hat u_{t,\uparrow}(r) \hat u_{t,\downarrow}(r') \int &\dd x\dd y\dd z \dd z'\, W(x-y) W(z-z') \ee^{\ii r \cdot (x-z)} \ee^{\ii r' \cdot (y-z')} \\ &   \times  a^*_\uparrow (v_x)a^*_\downarrow(v_y) a_\downarrow(v_{t,z'})a_\uparrow(v_{t,z})
	\end{align}
	we can bound $0\leq \hat u_{t,\sigma}(k) \leq \ee^{-t |k|^2}$ and thus 
	\begin{align}\nonumber
		|\langle\psi,\mathrm{I}^t_{6}\psi\rangle|  \le  \frac 1{L^4} \sum_{r,r'}& \ee^{-t (|r|^2 + |r'|^2)} \left\|  \int \dd x\dd y\, W(x-y) \ee^{-\ii r \cdot x} \ee^{-\ii r' \cdot y}  a_\downarrow(v_y) a_\uparrow (v_x) \psi \right\| 
		\\ \quad &\times \left\|  \int \dd z \dd z'\,  W(z-z') \ee^{-\ii r \cdot z} \ee^{-\ii r' \cdot z'}   a_\downarrow(v_{t,z'})a_\uparrow(v_{t,z})\psi \right\| .
	\end{align}
	The function $g_t$ with Fourier coefficients $\hat g_t(k) = \ee^{-t |k|^2}$ is the periodization of a Gaussian, hence non-negative and integrates to $1$, independently of $t$. 
	Using the Cauchy--Schwarz inequality for the sum over $r$ and $r'$, it suffices  to bound 
	\begin{align}\nonumber
		&   \frac 1{L^4} \sum_{r,r'} \ee^{-t (|r|^2 + |r'|^2)} \left\|  \int \dd x\dd y\, W(x-y) \ee^{-\ii r \cdot x} \ee^{-\ii r' \cdot y}  a_\downarrow(v_y) a_\uparrow (v_x) \psi \right\|^2 
		\\ \nonumber &=  \int \dd x\dd y \dd x' \dd y' \, W(x-y)W(x'-y') g_t(x-x') g_t(y-y') \left\langle \psi , a^*_\uparrow(v_{x'}) a^*_\downarrow(v_{y'}) a_\downarrow(v_y) a_\uparrow (v_x) \psi \right\rangle
		\\ & \leq  \|v_\downarrow\|_2^2  \int \dd x\dd y \dd x' \dd y' \, W(x-y)W(x'-y') g_t(x-x') g_t(y-y')   \| a_\uparrow(v_{x'})\psi\|  \| a_\uparrow (v_x) \psi \| \nonumber
		\\ & \leq  \|v_\downarrow\|_2^2   \| W* g_t \|_2^2 \langle \psi, \cN \psi\rangle
	\end{align}
	where the last step follows from $2 \| a_\uparrow(v_{x'})\psi\|  \| a_\uparrow (v_x) \psi \|  \leq   \| a_\uparrow(v_{x'})\psi\|^2 +  \| a_\uparrow(v_{x'})\psi\|^2$ and \eqref{4.59}.   An analogous bound holds with $v_t$ in place of $v$, and we conclude that
	\begin{align}\nonumber
	|\langle\psi,\mathrm{I}^t_{6}\psi\rangle| & \le 
		C \varrho\, \ee^{ t (k_{\rm F}^\downarrow)^2} \ee^{ t (k_{\rm F}^\uparrow)^2} \| W* g_t \|_2^2 \langle \psi, \cN \psi\rangle  \\ & = C \varrho\,  \langle \psi, \cN \psi\rangle \frac 1{L^2} \sum_k | \hat W(k)|^2 \ee^{-2t |k|^2} \ee^{ t (k_{\rm F}^\downarrow)^2} \ee^{ t (k_{\rm F}^\uparrow)^2} .
	\end{align}
	For the integration over $0< t <  k_{\rm F}^{-2}$ we shall use that 
	\begin{equation}\label{4.69}
	\int_0^{k_{\rm F}^{-2}} \dd t \, \ee^{-t (2 |k|^2 + 2\varepsilon - (k_{\rm F}^\downarrow)^2- (k_{\rm F}^\uparrow)^2) } 
	\leq \frac {C}{|k|^2 + k_{\rm F}^2}
	\end{equation}
	which can be obtained from an explicit calculation, using $\varepsilon>0$ and $k_{\rm F}^\sigma \leq k_{\rm F}$. 
	With $|\hat W(k)| \leq C \delta$ and \eqref{need:W}, we thus conclude that $\rm I_6$ satisfies the same bound as $\rm I_3$ in \eqref{I3_final}.  
	
	For the term $\mathrm I_7$, we find
	\begin{equation}
		\mathrm{I}^{t}_{7}= \int \dd x\dd y\dd z \dd z'\, W(x-y) W(z-z') v_{t,\downarrow}(y;z')u_{t,\uparrow}(x;z) a^*_\uparrow (v_x)a^*_\downarrow(u_y) a_\downarrow(u_{t,z'})a_\uparrow(v_{t,z}).
	\end{equation}
	It can be bounded as 
	\begin{equation}\label{eq:i 7 t less}
		|\langle\psi,\mathrm{I}^t_{7}\psi\rangle|\le  \left( \int W \right)^2   \|u_{t,\uparrow}\|_2 \|v_{t,\downarrow}\|_2  \| v_\uparrow\|_2^2  \langle \psi , \cN \psi\rangle 
	\end{equation}
	and hence satisfies \eqref{I3_final} after integration over $t$. The same applies to $\mathrm I_8$. 
	The term  $\mathrm I_9$ has a somewhat different structure, and can be written as
	\begin{equation}
		\mathrm{I}^{t}_{9}= \int \dd x\dd y\dd z \dd z'\, W(x-y) W(z-z') v_{t,\downarrow}(y;z')u_{t,\uparrow}(y;z') a^*_\uparrow (v_x)a^*_\uparrow(u_x) a_\uparrow(u_{t,z})a_\uparrow(v_{t,z}).
	\end{equation}
	With Cauchy--Schwarz, it can bounded in exactly the same was as \eqref{eq:i 7 t less}, however. 
	The same bound applies to $\mathrm I_{10}$, interchanging $\uparrow$ and $\downarrow$. 
	
	We are left with studying the quadratic terms which, by translation-invariance, are necessarily of the form $d\Gamma(\omega)$ for a  (spin-dependent) Fourier multiplier $\omega$. Explicitly, 
	\begin{equation}
	\mathrm I_{11} + \mathrm I_{12} + \mathrm I_{13} +\mathrm  I_{14} = \sum_{k,\sigma} \omega_\sigma(k)  \hat a_{k,\sigma}^* \hat a_{k,\sigma} 
	\end{equation}
	with 
	\begin{align}\nonumber
	\omega_\sigma(k) & = \frac 1{L^4} \sum_{p,r}  \hat W(p)^2 \left(  \frac{   \hat v_{\sigma'}(r) \hat u_{\sigma'}(r-p)  \hat v_\sigma(k- p) \hat u_\sigma(k) }{ |r-p|^2 - |r|^2 + |k|^2 - |k-p|^2+2 \varepsilon}  \right. \\ & \qquad \quad \left.  
	\qquad\qquad\quad	
	+ 
	 \frac{  \hat v_{\sigma'}(r) \hat u_{\sigma'}(r-p)  \hat u_\sigma(k+ p) \hat v_\sigma(k)  }{ |r-p|^2 - |r|^2 + |k+p|^2 - |k|^2 + 2 \varepsilon} \right)
	\end{align}
	(where $\sigma' \neq \sigma$). They can thus be bounded by $\cN\max_{k,\sigma} \omega_{\sigma}(k) $. Let us show that $\max_{k,\sigma} \omega_{\sigma}(k) \leq C \varrho \delta^2 \ln \delta^{-1}$. We can drop the positive terms $|k|^2 - |k-p|^2$ and $|k+p|^2 - |k|^2$ in the denominators, yielding
	\begin{equation}
	\omega_\sigma(k) \leq \frac 1{L^4} \sum_{p,r}  \hat W(p)^2  \frac{   \hat v_{\sigma'}(r) \hat u_{\sigma'}(r-p)  }{ |r-p|^2 - |r|^2 +2 \varepsilon}.
	\end{equation}
	By proceeding similarly  as in the proof of \Cref{lem:integration-t-new}, bounding the sum by the corresponding integral, and distinguishing between the cases $|p| \gtrsim k_{\rm F}^{\sigma'}$ and  $|p| \lesssim k_{\rm F}^{\sigma'}$,  one readily checks that 
	\begin{equation}
	 \frac 1{L^2} \sum_{r}  \frac{   \hat v_{\sigma'}(r) \hat u_{\sigma'}(r-p)  }{ |r-p|^2 - |r|^2 +2 \varepsilon} \leq  C \min\left\{ \frac {(k_{\rm F}^{\sigma'})^2}{|p|^2 + (k_{\rm F}^{\sigma'})^2} , \ln \delta^{-1}\right\}.
	\end{equation}
	In combination with \eqref{int:W} and \eqref{need:W}, this implies the result.
	
	The last term, $\mathrm I_{15}$, gives the constant first term on the right-hand side of \eqref{eq:prop conj q2}. This completes the proof. 
\end{proof}

\subsection{Conjugation of \(\bQ_4\)}\label{sec:q4}
In this subsection, we show that \(\bQ_4\) does not contribute to the energy to the relevant order of \(\delta^2\varrho^2\). We will follow analogous bounds in the three-dimensional case in \cite[Props.~4.5 \& 5.3]{giacomelli2024huangyangformulalowdensityfermi}, but again there will be notable differences. For later use, we shall formulate the next proposition for more general operators that have the same structure as $\bQ_4$,  but $W$ is replaced by a general function $w$. 

\begin{Prop}
\label[Prop]{prop:conj-q4}
Let \(\bQ_4|_w\) denote the same operator as \(\bQ_4\) in \eqref{eq: def H-corr}, but with \(W_\infty\) replaced by an even and integrable function \(w : \bR^2\to \bR_+\) with support contained inside \(\Lambda\). Then, for any \(\lambda\in[0,1]\),  
	\begin{equation}\label{eq:bound eq4}
		|\langle T_\lambda \Omega,{\bQ_4|_w}T_\lambda\Omega\rangle|\le 
		CL^2 \left( \int w \right) \varrho^2 \left( \delta \ln \delta^{-1} \right)^2\,.
	\end{equation}
\end{Prop}

\begin{proof}
As in \cite[Prop.~5.3]{giacomelli2024huangyangformulalowdensityfermi}, the first step will be to replace the functions $\hat u_\sigma$ in the definition of $\mathbb{Q}_4|_w$ 
by smooth functions with smaller support. Let again $k_{\rm F}= \max\{k_{\rm F}^\uparrow,k_{\rm F}^\downarrow\}$, and let $\hat\zeta^\gtrless: \mathbb{R}^2 \to [0,1]$ be smooth functions with $\hat \zeta^< + \hat \zeta^> = 1$ such that
\begin{equation}\label{def:zeta}
\hat \zeta^<(k) = \begin{cases} 1 & \text{if $|k| \leq 5 k_{\rm F}$} \\ 0 & \text{if $|k| \geq 6 k_{\rm F}$} \end{cases} 
\end{equation}
Let $\zeta^\gtrless$ be the corresponding functions on $\Lambda$ with Fourier coefficients $\hat \zeta^\gtrless$, and also $u_\sigma^\gtrless$ with Fourier coefficients $\hat u_\sigma^\gtrless(k) = \hat u_\sigma(k) \zeta^\gtrless(k)$. Note that $u_\sigma^> = \zeta^>$. Since $u_\sigma = u_\sigma^> + u_\sigma^<$, we can split
\begin{equation}
\mathbb{Q}_4|_w = \mathbb{Q}^>_4|_w + \mathbb{Q}^<_4|_w
\end{equation}
with 
\begin{equation}\label{4.80}
\mathbb{Q}^>_4|_w   = \frac{1}{2}\sum_{\sigma,\sigma^\prime}\int_{\Lambda^2} \dd x \dd y \, w(x-y)  {a}^\ast_\sigma(u^>_x){a}^\ast_{\sigma^\prime}(u^>_y){a}_{\sigma^\prime}(u^>_y){a}_\sigma(u^>_x)  
\end{equation}
and $\mathbb{Q}^<_4|_w$ containing all the remaining terms. Using $\|u_\sigma^<\|_2 \leq C \varrho^{\frac 12}$, the Cauchy--Schwarz inequality implies that 
\begin{equation}\label{4.62}
\mathbb{Q}^<_4|_w \leq \mathbb{Q}^>_4|_w + C \varrho \cN  \int w  .
\end{equation}
The second term on the right-hand side can be bounded with the aid of \Cref{prop:bound-n}, and we are hence left with bounding $\mathbb{Q}^>_4|_w$. 

We shall do this via Gr\"onwall's lemma, and start by computing the derivative, 
\begin{equation}
	     \partial_\lambda  T^\ast_{\lambda} \mathbb{Q}^>_4|_w T_{\lambda} 
		        =  T^\ast_{\lambda}[\mathbb{Q}^>_4|_w, B] T_{\lambda} + \mathrm{h.c.}
\end{equation}
By \Cref{def:bosonic-transformations} and \eqref{4.80},  we thus need to compute the commutator
\begin{equation}
[\hat{a}^\ast_{s+k,\sigma}\hat{a}^\ast_{s^\prime - k,\sigma^\prime}\hat{a}_{s^\prime, \sigma^\prime}\hat{a}_{s,\sigma},
                \hat{a}_{p+r,\uparrow}\hat{a}_{-r,\uparrow} \hat{a}_{-p+r',\downarrow}\hat{a}_{-r',\downarrow}  ] 
\end{equation}
in case $s,s+k \not\in B_\sigma$, $s', s'-k\not\in B_{\sigma'}$, $p+r \not \in B_\uparrow$, $r\in B_\uparrow$, $r'-p\not\in B_\downarrow$, $r'\in B_\downarrow$. After normal ordering, we obtain $6$ terms, given by
    \begin{align}\nonumber
         &  - \delta_{s +k, r+p}\delta_{\sigma, \uparrow}\hat{a}^\ast_{s^\prime -k,\sigma^\prime}\hat{a}_{r^\prime - p,\downarrow}\hat{a}_{-r^\prime, \downarrow}\hat{a}_{-r,\uparrow}\hat{a}_{s^\prime, \sigma^\prime}\hat{a}_{s,\sigma}
         +\delta_{s^\prime - k, r+p}\delta_{\sigma^\prime,\uparrow}\hat{a}_{s+k,\sigma}^\ast\hat{a}_{r^\prime - p, \downarrow}\hat{a}_{-r^\prime, \downarrow}\hat{a}_{-r,\uparrow}\hat{a}_{s^\prime, \sigma^\prime}\hat{a}_{s,\sigma} 
         \\ \nonumber & - \delta_{s^\prime -k, r^\prime -p}
       \delta_{\sigma^\prime, \downarrow}\hat{a}^\ast_{s+k,\sigma}\hat{a}_{r+p,\uparrow}\hat{a}_{-r^\prime, \downarrow}\hat{a}_{-r,\uparrow}\hat{a}_{s^\prime, \sigma^\prime}\hat{a}_{s,\sigma}
         + \delta_{s+k,r^\prime - p}\delta_{\sigma,\downarrow}\hat{a}^\ast_{s^\prime - k, \sigma^\prime}\hat{a}_{r+p,\uparrow}\hat{a}_{-r^\prime, \downarrow}\hat{a}_{-r,\uparrow}\hat{a}_{s^\prime, \sigma^\prime}\hat{a}_{s,\sigma} 
         \\
         &+\delta_{s+k,r+p}\delta_{s^\prime - k, r^\prime -p}\delta_{\sigma,\uparrow}\delta_{\sigma^\prime,\downarrow}\hat{a}_{-r^\prime, \downarrow}\hat{a}_{-r,\uparrow}\hat{a}_{s^\prime, \sigma^\prime}\hat{a}_{s,\sigma}
         - \delta_{s+k,r^\prime - p}\delta_{s^\prime -k,r+p}\delta_{\sigma,\downarrow}\delta_{\sigma^\prime, \uparrow}\hat{a}_{-r^\prime, \downarrow}\hat{a}_{-r,\uparrow}\hat{a}_{s^\prime, \sigma^\prime}\hat{a}_{s,\sigma} 
    \end{align}
	Exploiting the reflection symmetry of \(\hat w(k)\) and \(\hat \eta^\varepsilon_{r,r'}(p)\), we can combine 
	the first two terms,  the third and the forth as well as the last two
	by a suitable change of variables. 
	Accordingly, we shall write 
	\begin{equation}
		[\mathbb{Q}^>_4|_w, B]= \sum_{i=1}^3 \mathrm{J}_i \,.
	\end{equation}
	Similarly as in the proof of \Cref{prop:conj-bq2}, we shall denote by \(\mathrm J^t_j \) an expression such that \(\mathrm {J}_j =\int_0^\infty\dd t\, \ee^{-2t\varepsilon}\,\mathrm{J}^t_j \). 
	
	The first term to consider is 
	\begin{align}\nonumber
	\mathrm J_1^t & = \sum_\sigma \int \dd x \dd y \dd z \dd z' \, \zeta_{t}^>(z-z') w(x-z) W(y-z')  \\ & \qquad\qquad \times a_\sigma^*(u^>_x) a_\downarrow(u_{t,y}) a_\downarrow(v_{t,y}) a_\uparrow(v_{t,z'})  a_\uparrow(u^>_{z}) a_\sigma(u^>_{x})
	\end{align}
	where we introduced the function $\zeta^>_t$ with Fourier coefficients $\hat \zeta^>(k) \ee^{-t |k|^2}$. 
	We in fact have 
	\begin{equation}
	\mathrm J_1 =  \sum_\sigma \int \dd x  \dd z \,  w(x-z)  a_\sigma^*(u^>_x) \mathcal{O}_z   a_\uparrow(u^>_{z}) a_\sigma(u^>_{x})
	\end{equation}
	with
	\begin{align}
	\mathcal{O}_z &
	=  \frac 1{L^3} \sum_{p,r,s} \hat W(p) \frac{ \hat \zeta^>(p+s) \hat u_{\downarrow}(p-r) \hat v_{\downarrow}(r) \hat v_{\uparrow}(s)}{ |p-r|^2 + |p+s|^2 - |r|^2 - |s|^2 + 2 \varepsilon}  \hat  a_{p-r,\downarrow} \hat a_{r,\downarrow} \hat a_{s,\uparrow}  \ee^{\ii z \cdot(s+p)}.
	\end{align}
	Since $|p+s|\geq 5 k_{\rm F}$ for all the summands, and also $|r|\leq k_{\rm F}^\downarrow$ and $|s| \leq k_{\rm F}^\uparrow$, we have $|p-r| \geq 3 k_{\rm F}$, and hence we can replace $\hat u_\downarrow(p-r)$ by the characteristic function of $|p-r| \geq 3 k_{\rm F}$, which we shall denote by $\hat \eta$. For $t>0$, let us introduce 
	\begin{equation}\label{def:eta}
	\hat \eta_t (k ) = \begin{cases} \ee^{-t |k|^2} & \text{for $|k| \geq 3 k_{\rm F}$} \\ 0 & \text{for $ |k| < 3 k_{\rm F}$} \end{cases}
	\end{equation}
	 The conclusion is thus that we can equivalently write
	\begin{align}\nonumber
	\mathrm J_1^t & =  \sum_\sigma\int \dd x \dd y \dd z \dd z' \, \zeta_{t}^>(z-z') w(x-z) W(y-z') \\ & 
	\qquad\qquad \times a_\sigma^*(u^>_x) a_\downarrow(\eta_{t,y}) a_\downarrow(v_{t,y}) a_\uparrow(v_{t,z'})  a_\uparrow(u^>_{z}) a_\sigma(u^>_{x})
	\end{align}
	and hence
	\begin{equation}
	|\langle \psi, \mathrm J_1^t \psi \rangle|  \le \left( \int w\right)^{\frac 12}\left( \int W\right) \| \zeta_t^>\|_1 \|\eta_t\|_2 \| v_{t,\downarrow}\|_2 \| v_{t,\uparrow}\|_2  \|\bQ^>_4|_w^{\frac{1}{2}}\psi\| \|\cN^{\frac{1}{2}}\psi\| 
	\end{equation}
	by Cauchy--Schwarz. We have \(\|\zeta^>_t\|_1 \leq C\) uniformly in $t>0$, as shown in \cite[Lemma~A.6]{giacomelli2024huangyangformulalowdensityfermi}. Moreover, by proceeding as in the proof of \Cref{lem:integration-t-new}, one readily shows that 
\begin{equation}\label{inte4}
\int_0^\infty \dd t\,  \ee^{-2t\varepsilon}	\|\eta_t\|_2 \| v_{t,\sigma}\|_2 \ee^{t k_{\rm F}^2} \leq C  .
\end{equation}	
(Using $\|\eta_t\|_2 \leq \|u_{t,\sigma}\|_2$ an application of \eqref{eq:lem:integration-t-1-new} would lead to an additional factor $\ln \delta^{-1}$ on the right-hand side; this results from momenta close to the Fermi sphere, hence does not appear here since $\hat\eta$ is supported well outside the Fermi balls, however.) 
This leads to the bound  
\begin{equation}
	|\langle \psi, \mathrm J_1 \psi \rangle|  \le  C \varrho^{\frac 12} \left( \int w\right)^{\frac 12}\left( \int W\right) \|\bQ_4^>|_w^{\frac{1}{2}}\psi\| \|\cN^{\frac{1}{2}}\psi\|  .
	\end{equation} 
	The term $\mathrm J_2$ is actually the same, only with spins reversed. 

We are left with $\mathrm J_3$, for which we have 
\begin{equation}
\mathrm J_3^t \! =\! \int \dd x \dd y \dd z \dd z' \, W(x-y) w(z - z') \zeta^>_t (x-z) \zeta^>_t (y-z') a_\downarrow(v_{t,x}) a_\uparrow(v_{t,y}) a_\downarrow(u^>_z) a_\uparrow(u^>_{z'}) .
\end{equation} 
With Cauchy--Schwarz, we can bound it as
\begin{align}\nonumber
	|\langle \psi, \mathrm J_3^t \psi \rangle| & \le \left( \int w\right)^{\frac 12} L  \| v_{t,\downarrow}\|_2 \| v_{t,\uparrow}\|_2  \|\bQ^>_4|_w^{\frac{1}{2}}\psi\| \\ & \quad \times \max_{z,z'} \int  \dd x\dd y\, W(x-y) |\zeta_t^>(x-z)| | \zeta_t^>(y-z')|  \,. \label{J3t}
\end{align}
For the last factor, we shall give two different bounds, useful for large and small $t$, respectively. On the one hand, we can simply bound it is
\begin{equation}\label{4.76}
 \int  \dd x\dd y\, W(x-y) |\zeta_t^>(x-z)| | \zeta_t^>(y-z')|  \le \left( \int W \right) \| \zeta^>_t \|_2^2
\end{equation}
and the resulting expression can be integrated over $t \geq k_{\rm F}^{-2}$. In fact, by proceeding as in the proof of \Cref{lem:integration-t-new} one finds 
\begin{equation}\label{inte5}
\int_{k_{\rm F}^{-2}}^\infty \dd t\,  \ee^{-2t\varepsilon} \|\zeta^>_t\|_2^2 \ee^{2 t k_{\rm F}^2} \leq C .
\end{equation}	
On the other hand, we can write $\zeta_t^> = g_{t-s} * \zeta_{s}^> $ where $g_t$ denotes the function with Fourier coefficients $\hat g_t(k) = \ee^{-t |k|^2}$, as already introduced in the proof of \Cref{prop:conj-bq2}.  
Using that $g_t \geq 0$ and $\| \zeta^>_s\|_1 \leq C$ uniformly in $s>0$, we can bound 
\begin{equation}\label{4.77}
\int  \dd x\dd y\, W(x-y) |\zeta_t^>(x-z)| | \zeta_t^>(y-z') | \leq  \frac {C^2 } {L^2} \sum_p \hat W(p) \ee^{-2 t |p|^2} .
\end{equation}
The right-hand side can now be integrated over $0< t < k_{\rm F}^{-2}$. 
With \eqref{4.69} and \eqref{need:W}, this leads to the conclusion that
\begin{equation}
|\langle \psi, \mathrm J_3 \psi \rangle|  \le \left(\int w\right)^{\frac 12} L \varrho \| \bQ_4^> |_w^{\frac 12} \psi \|  \delta  \ln \delta^{-1}.
\end{equation}

By combining the estimates above with Gr\"onwall's lemma, \eqref{4.62} and \Cref{prop:bound-n}, this completes the proof.
\end{proof}

\subsection{Conclusion of \Cref{thm:bound small box}}\label{sec: conclusion of bound on small box}

We now collect all the bounds in the previous subsections and proceed with the proof of \Cref{thm:bound small box}. Recall the choice of the trial state in \eqref{def:Phi}. With \Cref{prop: fermionic transf}, \eqref{int:W}, \eqref{eq:err simpl corr}, \eqref{q320} as well as \Cref{prop:bound-n,prop:conj-q4}, 
we have
\begin{align}\nonumber
L^{-2} \langle \Phi, \cH_W \Phi\rangle & = 2\pi(\varrho_\uparrow^2+\varrho_\downarrow^2)+\hat{W}(0)\varrho_\uparrow\varrho_\downarrow + L^{-2} \left\langle T \Omega , \left( \bH_0 + \bQ^{\uparrow\downarrow}_2 \right) T \Omega \right\rangle 
\\
& \quad + O\left(\delta b^2\varrho^{3}\right)+O\left( {\varrho^{\frac{3}{2}}}{L^{-1}} \right) 
+ O \left( \varrho^2  \delta^3  \left(\ln \delta^{-1} \right)^2  \right)  \label{concl}.
\end{align}
For the conjugation of \(\bH_0\) and \(\bQ_2^{\uparrow\downarrow}\) with \(T\), we can apply  \Cref{prop:conj-h-0,prop:conj-bq2}, which imply that
	\begin{align}\nonumber
		\langle T\Omega,(\bH_0+\bQ_2^{\uparrow\downarrow})T\Omega\rangle 
		=&\int_0^1\dd \lambda \, \lambda\partial_\lambda \langle T_\lambda\Omega,\bQ_2^{\uparrow\downarrow}T_\lambda\Omega\rangle 
		+\int_0^1\dd \lambda\, \langle T_\lambda\Omega,\cE_{\bH_0}T_\lambda\Omega\rangle
		\\=& \nonumber
		-\frac{1}{L^4}\sum_{p,r,r'\in\Lambda^*}\hat W(p)
		 \hat \eta^\varepsilon_{r,r'}(p) \hat u_\uparrow(r+p)\hat u_\downarrow(r'-p)\hat v_\uparrow(r)\hat v_\downarrow(r')\\ &
		 +\int_0^1\dd \lambda\, \langle T_\lambda\Omega,(\cE_{\bH_0}+\lambda\cE_{\bQ_2})T_\lambda\Omega\rangle \label{sumtobr}
	\end{align}
with the error terms satisfying 
	\begin{equation}
		\int_0^1\dd \lambda\, \langle T_\lambda\Omega,(\cE_{\bH_0}+\cE_{\bQ_2^{\uparrow\downarrow}})T_\lambda\Omega\rangle\le C L^2 \varrho^2 \delta^{3} \left( \ln \delta^{-1} \right)^2 
	\end{equation}
for small $\delta$, using again \Cref{prop:bound-n,prop:conj-q4}.

The sum on the right-hand side of \eqref{sumtobr} can be replaced by an integral, up to an error that is certainly smaller than the \(O( \varrho^{\frac{3}{2}} L^{-1})\) already present on the right-hand side of \eqref{concl}, coming from the analogous replacement in the kinetic energy.  In the resulting integral, we still have to remove the $\varepsilon$ in order to arrive at the desired term. This is achieved in the following lemma.

\begin{Le}
\label[Le]{lem:integral gap}
	For \(\varepsilon =\delta\varrho\), we have that
	\begin{equation}
		\begin{split}
			\int & \dd p \dd  k \dd r \, \hat W(p)^2 \frac{\hat u_\uparrow(r+p)\hat u_\downarrow(k-p)\hat v_\uparrow(r)\hat v_\downarrow(k)}{|r+p|^2-|r|^2+|k-p|^2-|k|^2+2\varepsilon} \\
			&=    
			\int \dd p \dd  k \dd r\,  \hat W(p)^2 \frac{\hat u_\uparrow(r+p)\hat u_\downarrow(k-p)\hat v_\uparrow(r)\hat v_\downarrow(k)     }{|r+p|^2-|r|^2+|k-p|^2-|k|^2}                         
			     + O(\delta^3\varrho^2).
		\end{split}
	\end{equation}
\end{Le}
\begin{proof}
We emphasize that the integrals above are indeed finite due to the decay of \(\hat W(p)\) for large \(|p|\), compare with \eqref{need:W}. 
	With the notation \(\lambda_{r,p}=|r+p|^2-|r|^2\) 
	we have
	\begin{align}\nonumber
 &			\int \dd p \dd  k \dd r \,\hat W(p)^2 \left(\frac{\hat u_\uparrow(r+p)\hat u_\downarrow(k-p)\hat v_\uparrow(r)\hat v_\downarrow(k)}{\lambda_{r,p}+\lambda_{k,-p}+2\varepsilon}-\frac{\hat u_\uparrow(r+p)\hat u_\downarrow(k-p)\hat v_\uparrow(r)\hat v_\downarrow(k)}{\lambda_{r,p}+\lambda_{k,-p}}\right)\\
			& =
			2\varepsilon\int \dd p \dd  k \dd r\, \hat W(p)^2 \frac{\hat u_\uparrow(r+p)\hat u_\downarrow(k-p)\hat v_\uparrow(r)\hat v_\downarrow(k)}{(\lambda_{r,p}+\lambda_{k,-p})(\lambda_{r,p}+\lambda_{k,-p}+2\varepsilon)} 	.
	\end{align}
	On the right-hand side, we can use \(| \hat W(p)| \le C\delta\) and $\varepsilon>0$ for an upper bound. The result then follows from
	\begin{equation}
	\int \dd p \dd  k \dd r\,  \frac{\hat u_\uparrow(r+p)\hat u_\downarrow(k-p)\hat v_\uparrow(r)\hat v_\downarrow(k)}{(\lambda_{r,p}+\lambda_{k,-p})^2 } \leq C \varrho
	\end{equation} 
	which one can readily prove by 
	adapting the  proof in the three-dimensional case in  \cite[Lemma C.1]{giacomelli2024huangyangformulalowdensityfermi}. 
\end{proof}

Note that 
\begin{equation}
		\hat W(0)= \int_{\bR^2} W_\infty =\frac{4\pi}{\ln(b/a)}.
\end{equation}
Moreover, for small $b^2\varrho$
\begin{equation}	
\frac 1{(2\pi)^2}	\int_{|p|\ge 2\ee^{-\gamma} \varrho^{\frac 12}}\dd p\ \frac{\hat W(p)^2}{2|p|^2}=\frac{ 4\pi}{(\ln(b/a))^2} \left( - \ln(b \varrho^{\frac 12}) +O \left( b^2 \varrho\right)\right) 
\end{equation}
	where we used \eqref{eq:fourier transform of v}
 and the fact that
	\begin{equation}
		\int_x^\infty J_0(p)^2 \frac{\dd p}{p}= \ln(2/x)-\gamma+O(x^2)
	\end{equation}
as $x\to 0$, which follows from a combination of the identities $(9.1.12)$, $(11.1.20)$ and $(11.4.43)$ in  \cite{AS}.
Since 
\begin{equation}
\frac{1}{\ln(b/a)} + \frac{\ln(b\varrho^{\frac 12})}{(\ln(b/a))^2} = 	- \frac{1}{\ln(a\varrho^{\frac 12})}\left(1+O\left(\frac{\ln(b\varrho^{\frac 12})}{\ln(b/a)}\right)^2\right)
\end{equation}
we thus have 
	\begin{equation}\label{eq:scat len rep}
		\hat W(0)-\frac 1{(2\pi)^2}\int_{|p|\ge 2\ee^{-\gamma}/R}\dd p\ \frac{\hat W(p)^2}{2|p|^2}
		= - \frac{4\pi}{\ln(a\varrho^{\frac 12})}\left(1+O\left(\frac{\ln(b\varrho^{\frac 12})}{\ln(b/a)}\right)^2\right).
	\end{equation}

In conclusion, with the choice $b = \varrho^{-\frac 12} \delta^{\gamma_b}$ for $\gamma_b \geq 1$, 
\begin{equation}\label{ad0}
\frac 1{L^{2}} \langle \Phi, \cH_W \Phi\rangle = 2\pi(\varrho_\uparrow^2+\varrho_\downarrow^2)+ \frac{4\pi \varrho_\uparrow\varrho_\downarrow }{|\ln(a\varrho^{\frac 12})|}+O\left( {\varrho^{\frac{3}{2}}}{L^{-1}} \right) + \mathrm{A} 
+ O \left( \varrho^2  \delta^3  \left(\ln \delta^{-1} \right)^2  \right) 
\end{equation}
with $\mathrm{A}$ given by  
\begin{align}\nonumber
\mathrm{A} & = -
  \frac 1{(2\pi)^6}\int  \dd p \dd  k \dd r\, \hat W(p)^2 \hat v_\uparrow(r)\hat v_\downarrow(k)  \\ & \qquad \qquad\qquad  \times \left(\frac{\hat u_\uparrow(r+p)\hat u_\downarrow(k-p)}{|r+p|^2-|r|^2+|k-p|^2-|k|^2}
		-\frac{\mathbbm{1}_{\{|p|\ge 2\ee^{-\gamma}\varrho^{\frac 12}\}}}{2|p|^2}
		\right) \label{ad} .
\end{align}
Finally, to arrive at the desired result \eqref{eq:bound small box}, we shall use that for our choice of $b$ we have \(\hat W(0)=  8\pi\delta  + O (\delta^{2} \ln \delta^{-1}) \) and, as already used to obtain \eqref{eq:sum fermi ball},  \(\hat {W}(p)=\hat {W}(0)+O(\delta b^2 |p|^2)\). From this one readily concludes that the $\hat W(p)^2$ in \eqref{ad} can be replaced by $(8\pi \delta)^2$ with an error that is negligible compared to the one of order $\varrho^2 \delta^3 (\ln \delta^{-1})^2$ already present in \eqref{ad0}. A simple rescaling shows that the resulting expression equals ${16\pi} \delta^2 F(\varrho_\downarrow/\varrho_\uparrow)\varrho_\uparrow\varrho_\downarrow$ with $F$ defined in \eqref{eq:def F mein thm}. 
This completes the proof of \Cref{thm:bound small box}.

\section{Step 4. Bound on error terms from Step 2}\label{sec:errJ}

In this section we shall bound the expectation values in our trial state \eqref{def:Phi} of the error terms from the introduction of the Jastrow factor in \Cref{sec:3}.

Let $\mathcal{R}_2 =\sum_{1\le i<j\le N}  (1-f_{ij}^2)$ denote the error term in \Cref{lem:norm  jastrow}.  

\begin{Prop}
\label[Prop]{prop:conj-q4-new}
For the trial state $\Phi$ in \eqref{def:Phi}, 
\begin{equation}\label{eq:bound error small boxes jastrow}
		\langle\Phi, \cR_2\Phi\rangle \le C\| 1-f^2\|_1 L^2 \varrho^2.
	\end{equation}
\end{Prop}

\begin{proof}
Let $w = 1 - f^2$. With the notation introduced in \Cref{prop:conj-q4}, the Cauchy--Schwarz inequality and $\| v_\sigma\|_2 \leq C \varrho^{\frac 12}$, we can bound
\begin{equation}\label{CS2}
R \cR_2 R \leq C \left(   \left( \int w\right) \varrho^2 L^2 + \left( \int w\right) \varrho \cN +  \bQ_4|_w   \right).
\end{equation}
Hence the claim follows immediately from \Cref{prop:bound-n,prop:conj-q4}.
\end{proof}

Recall the definition \eqref{def:R3} of the three-body error term $\cR_3$   in \Cref{lem:restrsoft pot}.

\begin{Prop}\label[Prop]{lem:error small boxes jastrow}
	For the trial state $\Phi$ in \eqref{def:Phi}, 
	\begin{equation}\label{eq:bound error small boxes jastrow2}
		| \langle\Phi, \cR_3\Phi\rangle| \le C\|f\nabla f\|_1^2 L^2 \varrho^3.
	\end{equation}
\end{Prop}
\begin{proof}
	We start by taking absolute values and bound
	\begin{equation}
		|\cR_3|  \le  \sum_{i,j,k}'  w(x_i-x_k)w(x_i-x_j)
	\end{equation}
	with \(w(x) = |f(x)\nabla f(x)|\). Similarly as for \eqref{CS2}, we can use the Cauchy--Schwarz inequality to further bound this as 
\begin{equation}\label{CS3}
R \cR_3 R \leq C \left(   \left( \int w\right)^2 \varrho^3 L^2 + \left( \int w\right)^2 \varrho^2 \cN  +  \bQ_4|_{w*w} + \left( \int w\right) \bQ_4|_w + \bS_6  \right)
\end{equation}	
with
	\begin{equation}\label{def:s6}
		\bS_6\!=\!\sum_{\sigma,\sigma',\sigma''}
		\int_{\Lambda^3} \dd x\dd y\dd z\, w(x-y)w(x-z)
		a^*_\sigma(u_x)a^*_{\sigma'}(u_y)a^*_{\sigma''}(u_z)a_{\sigma''}(u_z)a_{\sigma}(u_y)a_{\sigma}(u_x).
	\end{equation}
The first term on the right-hand side of \eqref{CS3} is the main term, the expectation value of the second, third and forth term in the state $T \Omega$  can be bounded with the aid of \Cref{prop:bound-n,prop:conj-q4}. We are thus left with showing that 
	\begin{equation}\label{tos6}
		|\langle T\Omega,\bS_6T\Omega\rangle|\le C L^2\varrho^3 \left( \int w \right)^2.
	\end{equation}

In order to prove \eqref{tos6}, we shall proceed similarly as for \(\bQ_4\) in \Cref{prop:conj-q4}. 
As a first step, we shall replace the functions $u_\sigma$ appearing in \eqref{def:s6} by $u_\sigma^>$. They are defined via their Fourier coefficients $\hat u_\sigma^> (k) = \hat u_\sigma(k) (1-\hat\zeta^<(k))$ 
with $\hat\zeta^<(k)$ defined in \eqref{def:zeta}. Denoting the resulting expression by $\bS_6^>$, we can use the Cauchy--Schwarz inequality to show that 
$\bS_6 - \bS_6^>$ satisfies the same bound as in \eqref{CS3} above (with $\bS_6^>$ in place of $\bS_6$ on the right-hand side), and hence we only need to consider $\bS_6^>$. 
	
In the following we shall prove that 
	\begin{equation}
		|\partial_\lambda \langle T_{\lambda}\Omega,\bS_6^> T_{\lambda}\Omega\rangle|\le
C L \varrho^{\frac 32}  \left(\int w\right) \| ( \bS_6^>)^{\frac 12}T_\lambda\Omega\| 
	\end{equation}
from which the desired bound \eqref{tos6} readily follows via Gr\"onwall's Lemma. For the derivative, we need to compute the commutator $[\bS_6^>,B]$. As in the previous section, we shall do this is momentum space, which requires the computation of 
\begin{equation}
[ 			\hat a^*_{p+l+k,\sigma} \hat a^*_{q-k,\sigma'}  \hat a^*_{r-l,\sigma''}
			\hat a_{r,\sigma''} 	\hat a_{q,\sigma'} 		\hat a_{p,\sigma},
			\hat a_{s+k',\uparrow}		\hat a_{-k',\uparrow}
			\hat a_{-s+l',\downarrow}	\hat a_{-l',\downarrow}
		]
\end{equation}
for $l'\in B_\downarrow$, $k'\in B_\uparrow$, while all the other $8$ indices appearing are outside the respective Fermi ball. In normal order, this results in $12$ terms; explicitly, 
\begin{align}\label{eq:comm s6 b1}
		  & -\delta_{p+l+k,s+k'}\delta_{\sigma,\uparrow}
		\hat a^*_{q-k,\sigma'}  \hat a^*_{r-l,\sigma''}
		\hat a_{r,\sigma''} 	\hat a_{q,\sigma'} 		\hat a_{p,\sigma}
		\hat a_{-k',\uparrow}
		\hat a_{-s+l',\downarrow}	\hat a_{-l',\downarrow}         \nonumber                     \\
		  & +\delta_{q-k,s+k'}\delta_{\sigma',\uparrow }
		\hat a^*_{p+l+k,\sigma}  \hat a^*_{r-l,\sigma''}
		\hat a_{r,\sigma''} 	\hat a_{q,\sigma'} 		\hat a_{p,\sigma}
		\hat a_{-k',\uparrow}\nonumber
		\hat a_{-s+l',\downarrow}	\hat a_{-l',\downarrow}                                       \\
		  & -\delta_{r-l,s+k'}\delta_{\sigma'',\uparrow}
		\hat a^*_{p+l+k,\sigma} \hat a^*_{q-k,\sigma'}
		\hat a_{r,\sigma''} 	\hat a_{q,\sigma'} 		\hat a_{p,\sigma}
		\hat a_{-k',\uparrow}
		\hat a_{-s+l',\downarrow}	\hat a_{-l',\downarrow}                \nonumber              \\
		  & +\delta_{p+l+k,-s+l'}\delta_{\sigma,\uparrow}
		\hat a^*_{q-k,\sigma'}  \hat a^*_{r-l,\sigma''}
		\hat a_{r,\sigma''} 	\hat a_{q,\sigma'} 		\hat a_{p,\sigma}
		\hat a_{s+k',\uparrow}		\hat a_{-k',\uparrow}
		\hat a_{-l',\downarrow}                                                \nonumber            \\
		  & -\delta_{q-k,-s+l'}\delta_{\sigma',\uparrow }
		\hat a^*_{p+l+k,\sigma}   \hat a^*_{r-l,\sigma''}
		\hat a_{r,\sigma''} 	\hat a_{q,\sigma'} 		\hat a_{p,\sigma}
		\hat a_{s+k',\uparrow}		\hat a_{-k',\uparrow}
		\hat a_{-l',\downarrow}                                                \nonumber            \\
		  & +\delta_{r-l,-s+l'}\delta_{\sigma'',\uparrow}
		\hat a^*_{p+l+k,\sigma} \hat a^*_{q-k,\sigma'}
		\hat a_{r,\sigma''} 	\hat a_{q,\sigma'} 		\hat a_{p,\sigma}
		\hat a_{s+k',\uparrow}		\hat a_{-k',\uparrow}
		\hat a_{-l',\downarrow}                                                  \nonumber          \\
		  & - \delta_{p+l+k,s+k'}\delta_{\sigma,\uparrow}\delta_{q-k,-s+l'}\delta_{\sigma',\downarrow }
		\hat a^*_{r-l,\sigma''}
		\hat a_{r,\sigma''} 	\hat a_{q,\sigma'} 		\hat a_{p,\sigma}
		\hat a_{-k',\uparrow}
		\hat a_{-l',\downarrow}                                              \nonumber              \\
		  & + \delta_{p+l+k,s+k'}\delta_{\sigma,\uparrow}\delta_{r-l,-s+l'}\delta_{\sigma'',\downarrow}
		\hat a^*_{q-k,\sigma'}
		\hat a_{r,\sigma''} 	\hat a_{q,\sigma'} 		\hat a_{p,\sigma}
		\hat a_{-k',\uparrow}
		\hat a_{-l',\downarrow}                                                \nonumber            \\
		  & + \delta_{q-k,s+k'}\delta_{\sigma',\uparrow } \delta_{p+l+k,-s+l'}\delta_{\sigma,\downarrow}
		\hat a^*_{r-l,\sigma''}
		\hat a_{r,\sigma''} 	\hat a_{q,\sigma'} 		\hat a_{p,\sigma}
		\hat a_{-k',\uparrow}
		\hat a_{-l',\downarrow}                                                      \nonumber      \\
		  & - \delta_{q-k,s+k'}\delta_{\sigma',\uparrow } \delta_{r-l,-s+l'}\delta_{\sigma'',\downarrow}
		\hat a^*_{p+l+k,\sigma}
		\hat a_{r,\sigma''} 	\hat a_{q,\sigma'} 		\hat a_{p,\sigma}
		\hat a_{-k',\uparrow}
		\hat a_{-l',\downarrow}                                                    \nonumber        \\
		  & - \delta_{r-l,s+k'}\delta_{\sigma'',\uparrow}\delta_{p+l+k,-s+l'}\delta_{\sigma,\downarrow}
		\hat a^*_{q-k,\sigma'}
		\hat a_{r,\sigma''} 	\hat a_{q,\sigma'} 		\hat a_{p,\sigma}
		\hat a_{-k',\uparrow}
		\hat a_{-l',\downarrow}                                                 \nonumber           \\
		  & + \delta_{r-l,s+k'}\delta_{\sigma'',\uparrow} \delta_{q-k,-s+l'}\delta_{\sigma',\downarrow }
		\hat a^*_{p+l+k,\sigma}
		\hat a_{r,\sigma''} 	\hat a_{q,\sigma'} 		\hat a_{p,\sigma}
		\hat a_{-k',\uparrow}
		\hat a_{-l',\downarrow},
	\end{align}
and correspondingly
	\begin{equation}
		\partial_\lambda T_{\lambda}^*\bS_6^> T_{\lambda}= T_{1,\lambda}^*[\bS_6^> ,B]T_{\lambda}+\mathrm{h.c.}=\sum_{j=1}^{12} \mathrm{K}_j+\mathrm{h.c.}
	\end{equation}
The terms naturally fall into two groups (the first six and the last six terms, or order $8$ and order $6$, respectively) where the representatives of each group can all be bounded similarly. 
In fact, a simple change of variables shows that $\mathrm{K}_1 = \mathrm{K}_4$, and also $\mathrm{K}_2 = \mathrm{K}_3 = \mathrm{K}_5 = \mathrm{K}_6$. We shall demonstrate how to bound $\mathrm{K}_1$ and $\mathrm{K}_7$. As in the previous section, we  adopt the notation $\mathrm{K}_j^t$ for an expression such that $\mathrm{K}_j = \int_0^\infty \dd t\, \ee^{-2 t \varepsilon} \mathrm{K}_j^t$. 

For the first term, we have
	\begin{align}\nonumber
		\mathrm{K}_1^t & =
		\sum_{\sigma',\sigma''}
	\int \dd x\dd y\dd z \dd y'\dd z'\,
		w(x-y)w(x-z) \zeta^>_t(z'-x)W(y'-z')         \\
		             & \qquad\qquad\times
		a^{*}_{\sigma'} ( u_z^>)a^{*}_{\sigma''}( u_y^>)
		a_{\sigma''}( u_y^>)a_{\sigma'}(  u_z^>) a_{\uparrow}( u_x^>) a_\uparrow(v_{t,z'})
		a_\downarrow(\eta_{t, y'})a_\downarrow(v_{t,y'}) 
	\end{align}
	where, as in the proof of \Cref{{prop:conj-q4}}, we have introduced the function $\zeta^>_t$ with Fourier coefficients $\hat \zeta^>(k) \ee^{-t|k|^2}$ (see \eqref{def:zeta}) and $\eta_t$ is defined in \eqref{def:eta}.
	It can be bounded with the aid of the Cauchy--Schwarz inequality as 
	\begin{equation}
		|\langle\psi,\mathrm{K}_1^{t}\psi\rangle| \le C\varrho^{\frac{1}{2}}\ee^{t (k_\mathrm{F}^\uparrow)^2}
		\left( \int W\right)  \|\zeta^>_t\|_1\|v_{t,\downarrow}\|_2
		\|\eta_t\|_2
		\|( \bQ^>_4|_{w*w})^{\frac{1}{2}}\psi\|
		\|( \bS_6^>)^{\frac{1}{2}}\psi\| .
	\end{equation}
	Again $\|\zeta^>_t\|_1$ is bounded uniformly in $t>0$, and for the remaining $t$-integration we can use \eqref{inte4}. 
	Using also \eqref{eq:bound eq4} (which also holds for $\bQ_4^>$ in place of $\bQ_4$, as the proof demonstrates), we conclude that 
	\begin{equation}\label{eq:i1s6t2}
		|\langle T_\lambda \Omega ,\mathrm{K}_1 T_\lambda \Omega\rangle| \le C L  \varrho^{\frac{3}{2}} \left( \int w \right) \| (\bS_6^>)^{\frac{1}{2}}T_\lambda \Omega\| \delta^2  \ln \delta^{-1}.
	\end{equation}
	For the other terms of order $8$, \(\|( \bQ^>_4|_{w*w})^{\frac 12}\psi\|\) is replaced by \( (\int w)^{\frac 12}\|( \bQ^>_4|_{w})^{\frac 12}\psi\|\), which, in view of \Cref{prop:conj-q4}, does not change the final result; the bound \eqref{eq:i1s6t2} also holds for those terms. 

	Similarly, we have for $\mathrm{K}_7$ 
	\begin{align}\nonumber
		\mathrm{K}_7^{t} &= \sum_{\sigma''} \int \dd x\dd y\dd z\dd y'\dd z'\,
		w(x-y)w(x-z)  W(y'-z') \zeta^>_t(z'-z)\zeta^>_t(y'-x) \\
		                   & \qquad\qquad\times
		\hat a_{\sigma''}^*( u^>_y)a_{\sigma''}( u^>_y)a_\downarrow( u^>_z)a_\uparrow( u^>_x)
		a_\uparrow(v_{t,z'})a_\downarrow(v_{t,y'})
	\end{align}
	An application of the Cauchy--Schwarz  yields 
	\begin{align}\nonumber
		|\langle\psi,\mathrm{K}_7^t\psi\rangle| &  \le
		 \left( \int w\right) \| v_{t,\downarrow}\|_2 \| v_{t,\uparrow}\|_2  
		\|\cN^{\frac{1}{2}}\psi\| \| (\bS_6^>)^{\frac{1}{2}}\psi\|
			\\ &                                       \label{eq:i2s6t2}
		\quad \times \sup_{x,z \in \Lambda} \int \dd y'\dd z'\, W(y'-z') | \zeta^>_t(z'-z)||\zeta^>_t(y'-x)|. 
	\end{align}
	\Cref{prop:bound-n} can be used in order to bound $\|\cN^{\frac{1}{2}}\psi\|$.  
	The remaining terms are of the exact same form as for $\mathrm{J}_3^t$ in \eqref{J3t} and can thus be bounded in the same way, with the result that, after integration over $t$, 
	\begin{equation}
	|\langle T_\lambda\Omega, \mathrm K_7 T_\lambda \Omega \rangle|  \le C L \varrho^{\frac 32}  \left(\int w\right) \|(\bS_6^>)^{\frac{1}{2}}T_\lambda \Omega\|  \left( \delta  \ln \delta^{-1}\right)^2.
	\end{equation}

	The remaining terms of order $6$  can be treated the same way. 
	This completes the proof.
\end{proof}

\section{Step 5. Conclusion of \Cref{thm:main-res}}\label{sec:proof main res}
We now have all the tools in hand in order to give the 
proof of \Cref{thm:main-res}. In the first step in \Cref{sec:2}, we have shown that it suffices to consider finite systems of size $\ell$ instead of taking the thermodynamic limit, see Eq.~\eqref{eq:step1}. This results in two error terms for the energy density of the order $\varrho^2 (\varrho \ell^2)^{- \frac 13}$ and $\varrho^2 (a/\ell)^{\frac {\varepsilon_0}{1+\varepsilon_0}}$, respectively. 

In the second step in \Cref{sec:3}, we have shown that instead of the original interaction potential $V$, one can work instead with the soft potential $W$ (defined in \eqref{def:Vtilde}). According to \Cref{lem:restrsoft pot}, the price to pay is an additive error of order $\varrho^2 (\ell/b)^2 (a/b)^{\varepsilon_0}$, as well as one resulting from $\cR_3$, which is estimated in \Cref{lem:error small boxes jastrow}. For its application, we shall need that 
\begin{equation}
\| f \nabla f\|_1 \leq C \frac b{\ln(b/a)} 
\end{equation}
which was shown in \cite[Eq.~(2.14)]{LY2d}. Hence the resulting error term is of order $\varrho^3 b^2 \delta^2$. Moreover, according to \Cref{lem:norm  jastrow}, we also have to take into account the change in norm due to the introduction of the Jastrow factor. This was bounded in \Cref{prop:conj-q4-new}, and we shall need that
\begin{equation}
\| 1- f^2 \|_1 \leq C \frac {b^2}{\ln(b/a)} 
\end{equation}
which follows from \cite[Eq.~(2.12)]{LY2d}. This adds another  error term of order $\varrho^2 (b \ell \varrho)^2 \delta  $. 
 
We shall choose $\ell \sim \varrho^{-\frac 12} \delta^{-\gamma_\ell}$ and $b \sim \varrho^{-\frac 12} \delta^{\gamma_b}$. All the above mentioned error terms are thus bounded by $\varrho^2 \delta^3$ if $\gamma_\ell \geq \frac 92$ and $\gamma_b \geq \gamma_\ell + 1$.

Finally, we still have to consider the constraint of filled Fermi balls in \Cref{thm:bound small box}, i.e., the assumption that $N_\sigma = |B_\sigma|$.  In order to fulfill it, it may be necessary to add more particles to the system, leading to a slightly increased density. It is easy to see that resulting densities are bounded by $\varrho_\sigma ( 1 + O (N_\sigma^{-\frac 12}))$, leading to an error term of the order $\varrho^{3/2} \ell^{-1}$, which is negligible for our choice of $\ell$ above, however.

We hence obtain the desired bound in \eqref{eq:main thm} as  a consequence of \Cref{thm:bound small box}.


\bigskip

\noindent
{\it Acknowledgments.}  This work was partially funded by the Deutsche Forschungsgemeinschaft (DFG, German Research Foundation) via TRR 352 --- Project-ID 470903074.



\end{document}